\documentclass[journal]{IEEEtran}

\usepackage{amsmath}
\usepackage{amssymb}
\usepackage{algorithm}
\usepackage{algorithmicx}
\usepackage{algpseudocode}
\usepackage{subfigure}
\usepackage{acronym}
\usepackage{todonotes}
\usepackage{hyperref}
\usepackage{hhline}
\usepackage{paralist} 
\usepackage[flushleft]{threeparttable}
\usepackage{amsthm}

\newtheorem{proposition}{Proposition}

\makeatletter
\renewcommand{\ALG@beginalgorithmic}{\scriptsize}
\makeatother

\algtext*{EndLoop}
\algtext*{EndFor}
\algtext*{EndIf}
\algtext*{EndWhile}

\newcommand{\name}{BEH}
\newcommand{\mono}{MONO}
\newcommand{\charbe}{ESA}
\newcommand{\anbe}{EHA}

\acrodef{AC}{Alternating Current}
\acrodef{ESS}{energy storage system}
\acrodef{RF}{Radio Frequency}
\acrodef{RFID}{Radio Frequency Identification}
\acrodef{USB}{Universal Serial Bus}
\acrodef{WATS}{Wireless autonomous transducer systems}
\acrodef{WiFi}{wireless local area network}
\acrodef{WPT}{Wireless Power Transfer}
\acrodef{BLE}{Bluetooth Low Energy}
\acrodef{FEC}{Forward Error Correcting}
\acrodef{EIRP}{Effective Isotropic Radiated Power}
\acrodef{RSSI}{Received Signal Strength Indicator}
\acrodef{WFI}{Wait For Interrupt}
\acrodef{PRR}{Packet Reception Rate}

\begin{document}

\title{{\name}: Indoor Batteryless BLE Beacons using RF Energy Harvesting for Internet of Things}
\author{
\IEEEauthorblockN{$\text{Qingzhi Liu}^{*\dagger}$, $\text{Wieger IJntema}^{*}$, $\text{Anass Drif}^{*}$, $\text{Przemys{\l}aw Pawe{\l}czak}^{*}$, and $\text{Marco Zuniga}^{*}$} \\
\IEEEauthorblockA{${}^{*}$EEMCS, Delft University of Technology, The Netherlands\\ ${}^{\dagger}$INF, Wageningen University, The Netherlands\\Email: qingzhi.liu@wur.nl, \{p.pawelczak, m.a.zunigazamalloa\}@tudelft.nl, \{w.ijntema, a.drif\}@student.tudelft.nl}
}

\maketitle

\begin{abstract}
Indoor beacon system is a cornerstone of internet of things. 
Until now, most of the research effort has focused on achieving various applications based on beacon communication. 
However, the power consumption on beacon devices becomes the bottleneck for the large-scale deployments of indoor networks. 
On one hand, the size of beacon devices is required to be small enough for easy deployment, which further limits the size of battery. 
On the other hand, replacing the batteries of the beacon devices could lead to high maintenance costs. 
It is important to provide them with full energy autonomy. 
To tackle this problem we propose \name: an indoor beacon system aimed at operating perpetually without batteries. 
Our contributions are twofold. Firstly, we propose four methods to increase the utilization efficiency of harvested RF energy in the beacon system, by which the energy consumption level becomes low enough to fit within the energy harvesting budget. 
Secondly, we implement {\name} using Bluetooth Low Energy (BLE) and Radio Frequency (RF) energy harvesting devices, and test {\name} extensively in a laboratory environment. Our test results show that {\name} can enable perpetual lifetime operation of mobile and static BLE beacon devices in a (16\,m$^{\text{2}}$) room. The average value of packet reception rate achieves more than 99\% in the best case. 
\end{abstract}

\section{Introduction}
\label{sec:introduction}

Beacon communication is widely adopted as one of the key technologies for Internet of Things (IoT) applications requiring low energy consumption and low data rate. 
Among various beacon schemes, Bluetooth Low Energy (BLE) beacon is one of the most promising systems for indoor applications. 
Its main advantages are low power consumption and simple implementation. 
Integrating BLE beacon into IoT devices promotes various applications, including localization~\cite{zhuang2016smartphone}, proximity detection~\cite{faragher2014analysis}, activity sensing~\cite{cheraghi2017guidebeacon}, smart offices~\cite{collotta2015novel}, etc. 

Along with the large-scale deployment, 
the maintainability in the life cycle of BLE beacon system raises as an important evaluation index, including deploying, replacing battery, repairing, etc. 
In some large-scale scenarios, the cost for replacing the batteries of thousands of beacon nodes would be too high. 
For example, Amsterdam Schiphol Airport has roughly 10000 BLE beacons deployed to provide navigation services. 
In some other scenarios, such as smart infrastructures with beacon devices embedded in the materials themselves, battery replacement would be even impossible. 

The research community has recognized this energy autonomous challenge, 
meanwhile the area of Energy Harvesting (EH) starts to gain momentum~\cite{xie:2012:wcm}\cite{xiao2014wireless}. 
To mitigate the power supply issue, various energy harvesting techniques are integrated with IoT devices. 
However, most existing energy harvesting solutions rely on specific application scenarios. 
For example, 
solar panel is mainly used for harvesting outdoor sunlight or indoor lighting. 
Kinetic energy harvesting targets on mobile users or objects. 
These energy harvesting systems cannot supply continuous harvested energy for perpetual operation of devices. 
We need an energy harvesting solution for BLE beacon applications, 
which can (i) be deployed for indoor area; (ii) operate 24 hours without interrupting the daily life of users; (iii) supply enough energy for periodical BLE beacon communication. 
Among the existing EH techniques~\cite{xie:2012:wcm}, we chose the one based on radio frequency (RF). 

In this paper, we build an indoor batteryless BLE beacon system using RF energy harvesting named \textbf{{\name}}. 
Dedicated RF sources are deployed to radiate radio energy to the neighbor BLE beacon devices. 
The main aim of {\name} is that the static and mobile beacon devices rely on harvested RF energy to achieve beacon communication and proximity detection. 
Specifically, the mobile device initiates the beacon procedure by broadcasting a beacon request message to the neighbor static BLE beacon nodes. 
The static beacon nodes send beacon reply messages back to the mobile devices after receiving the request. 
We require: 
(i) the mobile device can successfully receive the reply messages from static beacon nodes, and correctly decode the packet information. 
(ii) the mobile device localizes itself to the proximity position of the beacon node sending the beacon reply with the strongest signal strength. 

To achieve these aims, we make four innovations in the system design and implementation. 
Firstly, we propose a \emph{collision based beacon} approach based on capture effect and orthogonal code. The approach minimizes the power consumption of receiving beacon messages in mobile devices. 
Secondly, RF-based \emph{passive wake-up} is utilized and optimized to decrease the power consumption of idle listening and unnecessary wake-up in static beacon nodes. 
Thirdly, we make RF energy based \emph{proximity range estimation} to mobile nodes, and wakes up the static beacon nodes which are nearby the mobile device to send beacon reply messages. The approach avoids unnecessary wake-up and beacon communication in static beacon nodes. 
Fourthly, we propose a \emph{two-wave beacons} approach to cope with the unbalanced harvested energy in static beacon nodes. The approach increases the energy utilization efficiency of static beacon nodes that harvest lower energy. 

{\name} system can be used for sending information, such as sensing data, from static beacon nodes to mobile devices. 
The mobile devices can estimate its proximity position while receiving beacon information. 
We evaluate the performance of {\name} by packet reception rate (PRR) and proximity detection accuracy (PDA). 
We aim to achieve high PRR, while relax the requirement on PDA. We require PDA can provide mobile devices at least room level accuracy. The trade-off between PRR and PDA is not the aim of {\name}. 

To summarize, we make the following contributions. 

\begin{enumerate}
	\item To the best of our knowledge, {\name} is the first IoT system that successfully achieves bidirectional BLE beacon communication and proximity detection using the RF based energy harvesting. 
	\item We design and implement {\name} based on off-the-shelf BLE module~\cite{nrf51822}, RF energy transmitters and RF energy harvesters~\cite{tx91501_p2110}, which makes the fast adoption and deployment possible.
	\item  We systematically evaluate the system in an office environment. As a result, {\name} achieves perpetual beacon communication and proximity detection on batteryless devices. 
Some experimental results are shown in Tab.\ref{tab:test_results_example}. 
\end{enumerate}

\begin{table}
\center
\caption{Selected Experimental Results of BLE Beacon Applications based on RF Energy Harvesting}
\begin{threeparttable}
\begin{tabular}{l|l}
\hline
\textbf{Property} & \textbf{Test Results} \\ \hline \hline
Min. beacon period & 1 s \\ \hline
Max. density of static beacon nodes & 0.25 node per m$^2$ \\ \hline
Max. range from energy source to beacon devices & 4 m \\ \hline
Avg. PRR (0.0625 static beacon node per m$^{\text{2}}$) & $\approx 90\%$ \\ \hline
Avg. PDA (0.0625 static beacon node per m$^{\text{2}}$) & $\approx 94\%$ \\ \hline
Avg. PRR (0.25 static beacon node per m$^{\text{2}}$) & $\approx 99\%$ \\ \hline
Avg. PDA (0.25 static beacon node per m$^{\text{2}}$) & $\approx 71\%$ \\ \hline 
\end{tabular}
\begin{tablenotes}
\scriptsize \item $\S$ Tx Power of BLE is 0 dBm. Tx Power of RF energy transmitter is 3 W(EIRP). The antenna of RF energy harvester is vertical polarized patch antenna with 6.1\,dBi gain. 
\end{tablenotes}
\label{tab:test_results_example}
\end{threeparttable}
\end{table}

The rest of the paper is organized as follows. 
The background of RF based energy harvesting is introduced in Sec.~\ref{sec:background}.
The related work is discussed in Sec.~\ref{sec:realted_work}. 
The problem and solution is discussed in Sec.~\ref{sec:problem_statement}. 
The system model, including the hardware definition, operation procedure and evaluation metrics are listed in Sec.~\ref{sec:system_model}. 
In Sec.~\ref{sec:key_tech}, 
we explain the four key components of our system design in detail. 
We present our implementation of software and hardware in Sec.~\ref{sec:system_implementation}. 
Experimental results and analysis are presented in Sec.~\ref{sec:results}. 
Finally, we present the future work and our conclusions in Sec.~\ref{sec:discussion} and Sec.~\ref{sec:conclusion}.

\section{Background of RF Energy Harvesting}
\label{sec:background}

\begin{figure}
\centering
\includegraphics[width=0.9\columnwidth]{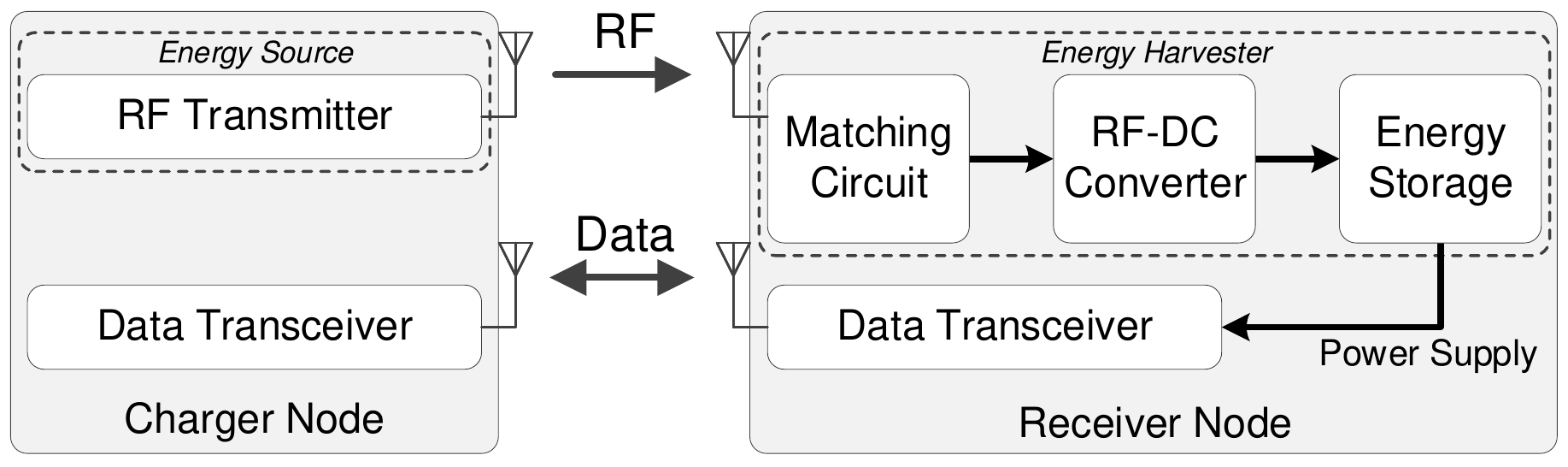}
\caption{The architecture of RF energy harvesting and BLE data communication in {\name}.} 
\label{pic:wpt_process}
\end{figure}

\begin{table}
\center
\caption{Pin functional description of P2110 $–$ 915 MHz RF Powerharvester${}^{\text{TM}}$ Receiver.}
\begin{threeparttable}
\begin{tabular}{ l | l }
\hline 
\textbf{Pin Label} & \textbf{Function} \\ \hline \hline
$\text{V}_{\text{CAP}}$ & Voltage of harvested RF energy after RF-DC converter. \\ \hline
$\text{D}_{\text{OUT}}$ & Analog voltage level corresponding to harvested power. \\ \hline
$\text{V}_{\text{OUT}}$ & DC Output of harvested power. \\ \hline
$\text{D}_{\text{SET}}$ & Set to enable measuring harvested power. \\ \hline
$\text{V}_{\text{THR}}$ & Threshold value of $\text{V}_{\text{CAP}}$ for charging and discharging. \\ \hline
\end{tabular}
\label{tab:powercast_pin}
\end{threeparttable}
\end{table}

\begin{figure}
\centering
\includegraphics[width=0.7\columnwidth]{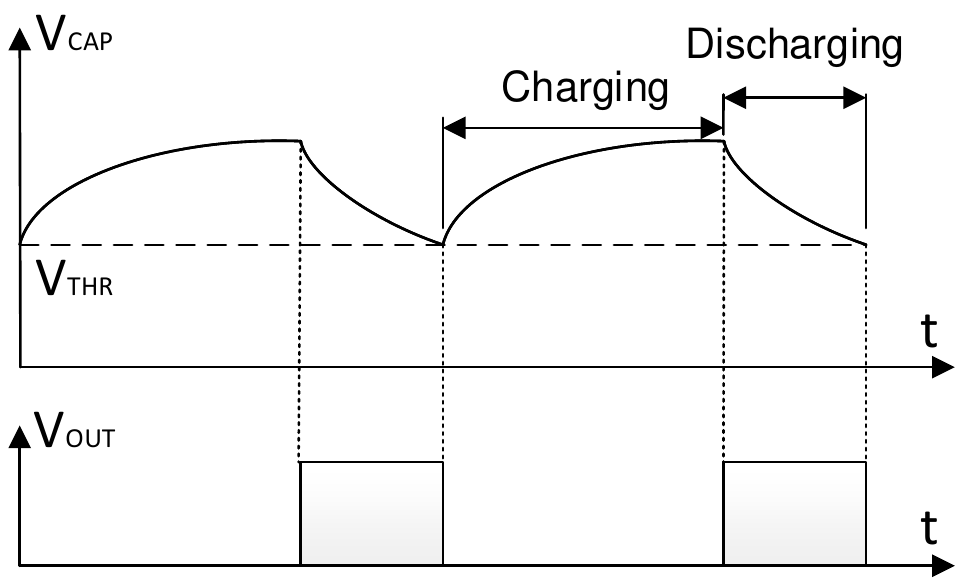}
\caption{The charging and discharging process of energy harvester. $\text{V}_{\text{OUT}}$ is the voltage supply for BLE beacon device.} 
\label{pic:powercast_charging_process}
\end{figure}

The basic architecture of RF energy harvesting and data communication used in {\name} is shown in Fig.\ref{pic:wpt_process}. The data communication channel is in parallel with RF energy harvesting. The energy used for the data transceiver in receiver node is from energy harvesting. The
core component in the energy harvester is RF-DC converter. It is used to rectify radio wave into directional current. Generally, it consists of band-pass filter, rectifier, and low-pass filter~\cite{jabbar2010rf}. 
The energy harvester boosts the voltage of the harvested power to the required level, and stores in capacitor for power supply of data transceiver. 

We utilize P2110 Powerharvester of Powercast~\cite{tx91501_p2110} as the RF energy harvester. 
There are five pins of P2110 used for {\name} as shown in Tab.~\ref{tab:powercast_pin}. 
When the voltage of harvested power from $\text{V}_{\text{CAP}}$ is above the threshold $\text{V}_{\text{THR}}$ in the capacitor, the voltage is boosted and the voltage output $\text{V}_{\text{OUT}}$ is enabled. 
When $\text{V}_{\text{CAP}}$ declines below the threshold $\text{V}_{\text{THR}}$, $\text{V}_{\text{OUT}}$ is turned off. 
A brief time diagram of the charging and discharging process is shown in Fig.\ref{pic:powercast_charging_process}. 
The time duration and frequency when $\text{V}_{\text{OUT}}$ has voltage output depends on the speed of charging and discharging. 
$\text{D}_{\text{OUT}}$ samples the analog voltage signal of harvested power to provide an indication of the amount of energy being harvested. 

The harvested RF energy is limited and easily affected by environmental factors, such as interference, obstacles, etc~\cite{liu2016safe}. 
To the best of our knowledge, there is not simulation model, which can accurately simulate all the factors affecting RF energy harvesting. 
Therefore, we implement {\name} in hardware and test in real deployment fields. 

\section{Related Work}
\label{sec:realted_work}

IoT system is typically used for industry automation, healthcare, surveillance, smart cities, etc. 
One of the bottlenecks for large-scale deployment is supplying adequate energy to IoT devices for long or even perpetual lifetime. 
There is much research about energy utilization efficiency and energy autonomous for IoT systems. 

\subsection{RF-based Energy Harvesting}

To mitigate the power supply issues for the large-scale deployment of IoT devices, 
energy autonomous for IoT attracts much research work~\cite{kamalinejad2015wireless}. 
In the research field of energy-autonomous systems, 
many kinds of techniques are explored as energy sources, including solar panels~\cite{raghunathan2005design}, wind~\cite{kansal2007power}, vibration~\cite{beeby2006energy}, RF radio~\cite{liu2016green}, magnetic resonance~\cite{kurs2007wireless}, ultra-sound~\cite{ubeam_website}, etc. 

RF-based energy harvesting systems can be classified into subcategories~\cite{xiao2014wireless}, including omni-directional or directional energy harvesting antenna, single or multi hop wireless charging~\cite{ju2014user}, static or mobile charger networks~\cite{peng2010prolonging}, parallel or simultaneous wireless information and power transfer~\cite{in2016blisp}~\cite{park2014joint}, ambient or dedicated RF energy harvesting~\cite{liu2016green}, etc. 

To implement an RF-based energy harvesting system, the first and most important factor is the type of RF energy source: ambient or dedicated RF energy harvesting. 
In the case of ambient RF energy source, the devices harvest ambient RF energy. The source of RF energy is not deployed purposely for energy harvesting, such as TV tower, wifi router, mobile base station, etc. 
In the case of dedicated RF energy harvesting, dedicated RF sources are deployed to transmit radio energy to the area where IoT devices require wireless energy~\cite{liu2016green}. 

The advantage of ambient RF energy harvesting is that there is no need to deploy specific energy sources; therefore, the deployment range of the energy harvesting devices is large. 
Compared with ambient RF energy harvesting, dedicated RF energy harvesting allows nearby devices to harvest more power. Because the energy source is dedicated to transmit RF power, and the radio properties of RF energy sources can be tuned to adapt to the energy requirements of energy harvesters, such as Tx power, frequency, etc. 
However, the disadvantage is that energy harvesting devices must be kept inside the effective power transmission range of the energy sources, which limits the mobility range of energy harvesting devices. 
In addition, the setup of RF energy sources must guarantee that the available RF power beamed on the mobile users is below the regulated safety bound~\cite{liu2016safe}. 
Based on our requirements on the BLE beacon applications, RF-based energy harvesting with dedicated RF energy sources is the most suitable choice. 

\subsection{BLE Beacon Applications}

BLE system is widely adopted for resource constrained IoT applications, such as indoor localization and tracking~\cite{zhuang2016smartphone}, proximity detection~\cite{faragher2014analysis}, activity sensing~\cite{cheraghi2017guidebeacon}, smart offices~\cite{collotta2015novel}, and so on. 
One of the main reasons is that BLE beacon communication consumes low energy and is simple to implement~\cite{jeon2018ble}. 
Therefore, BLE beacon becomes an ideal communication protocol to use energy harvesting. 
Much research combines BLE beacon based applications with energy harvesting techniques. 
\cite{nasiri2009indoor} analyzes the systems that harvest lighting energy for indoor energy constraint applications. 
The analysis is based on various design choices,  including photovoltaic cell, power conditioning circuit, energy storage, etc. 
\cite{wang2010design} designs and implements an indoor wireless sensor system using lighting-based energy harvesting. 
It proposes maximum power point tracking circuit, which increases the power conversion efficiency in solar cells.
The lifetime of the system largely increases from less than 6 months to more than 10 years. 
\cite{shih2016batteryless} proposes a BLE based hardware design that operates entirely on harvested energy. 
Its energy mainly relies on dual ISM-band RF sources and partially on photovoltaic energy harvesting. 
The system achieves periodical beacon advertising with interval of 45s. 

\subsection{Localization with RFID}

Radio-frequency identification (RFID) tags harvest RF energy from the electromagnetic field of nearby RFID readers. 
It is widely used to automatically identify and track tags. 
Localization based on RFID technology is developing rapidly in recent years. 
These approaches can be classified into three categories~\cite{ni2011rfid}.
In the first category, RFID reader-based localization system, e.g.~\cite{zhu2014fault, saab2011standalone}, allocates RFID reader in the object requiring localization to detect the pre-deployed anchor tags nearby. 
Although the deployment and maintenance cost of RFID tags are low, the localization lifetime depends on the limited battery of mobile readers.
In the second category, RFID tag-based localization system, e.g. \cite{yang2014tagoram, ni2004landmarc}, tracks the location of object attached with RFID tag by pre-deployed readers. 
The advantage of this approach is that the lifetime of tags is unlimited. 
However, the localization range is limited by the density of readers. 
Also, the tags cannot be localized once they are outside the range of readers. 
In the last category, RFID device-free localization system, e.g. \cite{liu2012mining}, detects the position of target wearing no additional localization devices. 
The idea is to find the target location by detecting and comparing the change of RFID signal in the environment. 

\subsection{Computational RFID (C-RFID)}

C-RFID is a RFID system with computational unit on RFID tag. 
It is applied in various wireless sensor applications~\cite{4519381}~\cite{1401841}. 
Wireless Identification and Sensing Platform (WISP)~\cite{4539485} is the most mature system among all the C-RFID research.
It uses the industrial standard EPCglobal Class 1 Generation 2 (EPC C1G2) RFID protocol. 
The main drawback of C-RFID is its intermittent power supply~\cite{zhang2012blink}, which drags down the computing and communication performance of the entire system. 

\begin{table}
\center
\caption{Comparison between BLE, RFID, C-RFID, and {\name}.}
\begin{tabular}{| l | l | l | l | l |}
\hline 
\textbf{Property} & \textbf{BLE} & \textbf{RFID} & \textbf{C-RFID} & \textbf{BEH} \\ \hline \hline
Data Processing & Yes & No & Yes & Yes \\ \hline
Battery & Yes & No & No & No \\ \hline
Energy Harvesting Source & No & Yes & Yes & Yes \\ \hline
Backscatter Radio & No & Yes & Yes & No \\ \hline
High Bound of Data Rate & High & Low & Low & High \\ \hline
\end{tabular}
\label{tab:comparison}
\end{table}

\subsection{Comparison}

We compare the properties between BLE, RFID, C-RFID, and our {\name} as shown in Tab.\ref{tab:comparison}. 
Compared with BLE, {\name} achieves BLE beacon operation without battery. 
The shortcoming of {\name} is mainly twofold. Firstly, {\name} system requires RF energy transmitters are deployed around the {\name} beacon devices, which limits the scalability of deployment. 
Secondly, the communication speed of {\name} depends on the amount of harvested energy. 
C-RFID and {\name} both have data processing unit, batteryless structure, and energy harvesting source. 
Compared with C-RFID, the key difference of {\name} is that its communication does not rely on backscatter radio. 
The communication and energy harvesting of {\name} use different radio frequencies. 
Intermittent energy radio only affects the harvested energy and communication speed of {\name}, while its communication quality is comparable to BLE communication as long as the harvested energy is enough. 
In addition, when the harvested energy is high, e.g. the {\name} device is close to RF energy source, the high bound of data rate using {\name} is much higher than RFID and C-RFID. 
Moreover, since {\name} uses BLE protocol stack for communication, it can be naturally integrated into existing BLE networks, and it can easily control the communication parameters, such as Tx power, to optimize the communication performance. 

\section{Problem Analysis}
\label{sec:problem_statement}

Beacon applications and energy harvesting are well researched topics on their own, but using RF-based energy harvesting for BLE beacon applications entails a substantial challenge. 
The main problem is that the amount of harvested energy from RF energy source is too small for the operation of beacon applications. 
In this section, 
we make empirical studies to analyze the challenges and possible solutions to achieve BLE beacon applications using RF-based energy harvesting. 

\subsection{Choose RF-based Energy Harvesting}
\label{sec:choose_rf_eh}

While many energy harvesting techniques exist~\cite{xie:2012:wcm}, such as solar energy harvesting, kinetic energy harvesting from human movement, indoor light harvesting, etc., we chose the one based on RF energy harvesting (RF-EH) with dedicated energy source. 
In this type of system, dedicated RF energy transmitters are deployed at specific positions. They continuously radiate radio energy to the neighbor RF energy harvesting devices. 
Compared with the other EH solutions, RF-EH with dedicated energy source has the following advantages. 

\begin{itemize}
	\item The operation of dedicated RF energy source does not rely on environment or users. For example, although indoor lighting harvesting is feasible for indoor BLE beacon applications, 
the harvested energy depends on the environmental lighting strength. 
	\item The transmitted power of dedicated RF energy source can be adjusted based on the requirement. The harvested energy can be estimated based on the distance between RF energy transmitters and harvesters. 
	\item To harvest the required amount of energy, the size of the RF energy harvesting device is small enough to attach on beacon devices. 
	\item Proximity detection is an important function in beacon applications. RF energy source can not only serve for transmitting RF energy, but also provide distance estimation between energy transmitter and receiver based on the harvested energy. This property can help to improve the performance of beacon based proximity detection. 
\end{itemize}

\begin{table}
\center
\caption{Electric Current under various operations in the datasheet of BLE nRF51822 Beacon Device~\cite{nrf51822_nRFgo}.}
\begin{threeparttable}
\begin{tabular}{l | l}
\hline
\textbf{Operation} & \textbf{Current} \\ \hline \hline
Standby (ON mode) & 2.6 $\mu$A \\ \hline
Run code from RAM at 16 MHz & 2.4 mA \\ \hline
Transmitting (4dBm) & 11.8 mA \\ \hline
Receiving & 9.7 mA \\ \hline
\end{tabular}
\begin{tablenotes}
\scriptsize \item $\S$ Voltage supply range from 1.8 V to 3.6 V. 
\end{tablenotes}
\label{tab:ble_power_consumption}
\end{threeparttable}
\end{table}

\subsection{Energy Requirement and Energy Budget}

To understand the challenges on BLE beacons using harvested RF energy, 
we measure and analyze the basic performance of BLE beacons and RF energy harvesting implemented by off-the-shelf devices. 

\subsubsection{Energy Consumption}
To assess the energy demand of BLE beacon devices, we use BLE beacon device nRF51822~\cite{nrf51822_nRFgo}. 
In the first place, we analyze the power consumption of its basic operations as shown in Tab.\ref{tab:ble_power_consumption}. 
As expected, the current on data transmitting and receiving is much higher than the other operations. 
The energy consumption on receiving is 3731 times larger than standby mode under the same voltage supply and time duration. 
This means that it will save much energy by setting the idle listening to standby mode. 

To further understand the energy consumption of BLE devices, 
we use Monsoon power meter~\cite{monsoon_website} to measure the power consumption and the time duration of sending and receiving BLE packets. 
Firstly, we set up an auxiliary BLE device. It listens to the communication channel in full time. Once receiving a packet, it sends back a packet of the same size. 
Secondly, we set up a measuring BLE device. 
It wakes up from sleeping mode and sends a packet. Then it listens to the communication channel and receives a packet from the auxiliary device. After receiving the packet, it returns to sleep mode. 
Fig~\ref{pre:energy_consume} shows the energy consumption of the measuring BLE device in various packet sizes. 

\subsubsection{Harvested RF Energy}

To assess the energy supply, we measure the harvested energy using off-the-shelf RF-based wireless power transfer product Powercast~\cite{tx91501_p2110}. 
The RF energy transmitter has an \ac{EIRP} of 3\,W and operates at a center frequency of 915\,MHz. 
We measure the harvested power (before the voltage booster) of Powercast P2110 receiver~\cite{p2110b_evb} with a 0.5\,k$\Omega$ load at various distances. Fig~\ref{pre:energy_harvest} shows the harvested energy in 100\,ms. 

\subsubsection{Analysis on Energy}

Based on the preliminary experiments above, we have the following observation results. 

\begin{enumerate}[(i)]
\item At a 1\,m distance from the RF energy transmitter, more than 350$\,\mu$J of energy can be harvested within 100\,ms, and the energy consumption from 10 to 100 bytes is no more than 40\,$\mu$J. 
Therefore, the harvested energy of 100\,ms gives sufficient leeway to transmit and receive beacon messages of 40 bytes in a range of 3 meters. 
\item The harvested energy within 100\,ms is not enough for transmitting and receiving packets of one byte at a distance 3.5\,m. 
This issue can be solved by increasing the time duration of RF energy harvesting. 

\item 
Due to the exponential decay of RF signal, the harvested energy at 1\,m is 36 times larger than at 3.5\,m. 
If beacon nodes with energy harvester are deployed uniformly around the RF energy transmitter, 
the harvested energy in these beacon nodes will be quite imbalanced. 
If all these nodes are required to send beacon messages in the same rate, they must wait until the last node harvests enough energy. 
Therefore, we need a solution to increase the energy utilization efficiency of the beacon devices that are further away from the RF energy transmitter. 
\end{enumerate}

\begin{figure}
\centering
\subfigure[BLE Mote nRF51822~\cite{nrf51822_nRFgo}]{\includegraphics[height=0.35\columnwidth]{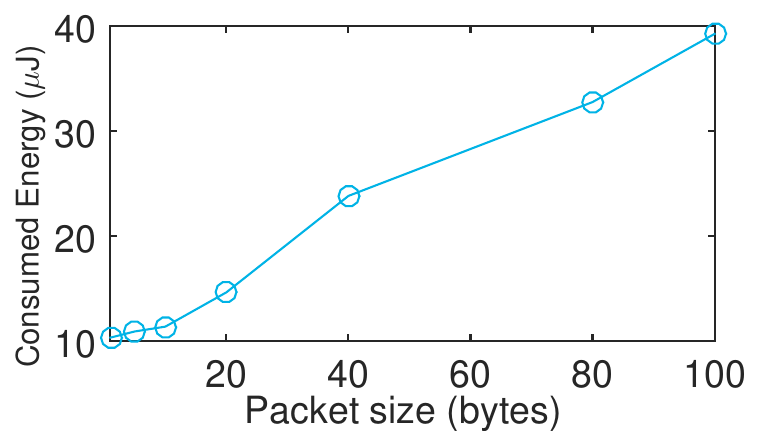}\label{pre:energy_consume}}
\subfigure[Powercast P2110~\cite{p2110b_evb}]{\includegraphics[height=0.35\columnwidth]{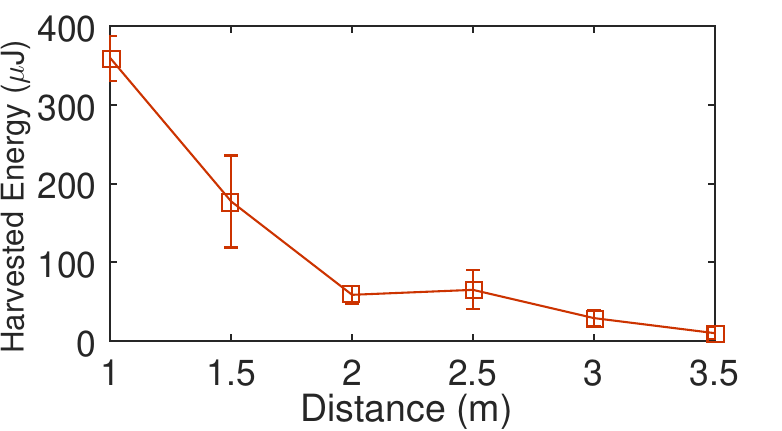}\label{pre:energy_harvest}}
\caption{Energy consumption and harvested RF energy in preliminary tests.}
\label{fig:energy_harvest_consume}
\end{figure}

\subsection{Solution Direction}

Based on the analysis above, 
we find if the harvested RF energy is used efficiently, it is enough to support the power consumption of BLE beacon applications. 
The key solutions to implement BLE beacon applications using RF energy harvesting are mainly two aspects: 

\begin{itemize}
\item \textit{Decreasing the power consumption of BLE beacon:}  
We must decrease the time length of beacon communication on both sender and receiver. 
The radio of beacon devices should be kept \emph{ON} for the minimum amount of time. 
The beacon receiver must avoid idle listening as much as possible. 

\item \textit{Increasing the efficiency of using harvested RF energy:} The beacon request rate must be carefully selected, so that the beacon nodes have enough harvested energy to finish the operation of sending and receiving beacon messages. 
The beacon devices must decrease the chance of unnecessary wakeup. 
Since the harvested energy in beacon devices is uneven, we must increase the energy utilization efficiency especially in the ``energy-weakest'' beacon nodes.  
\end{itemize}

\section{System Model}
\label{sec:system_model}

\begin{figure}
\centering
\includegraphics[width=0.8\columnwidth]{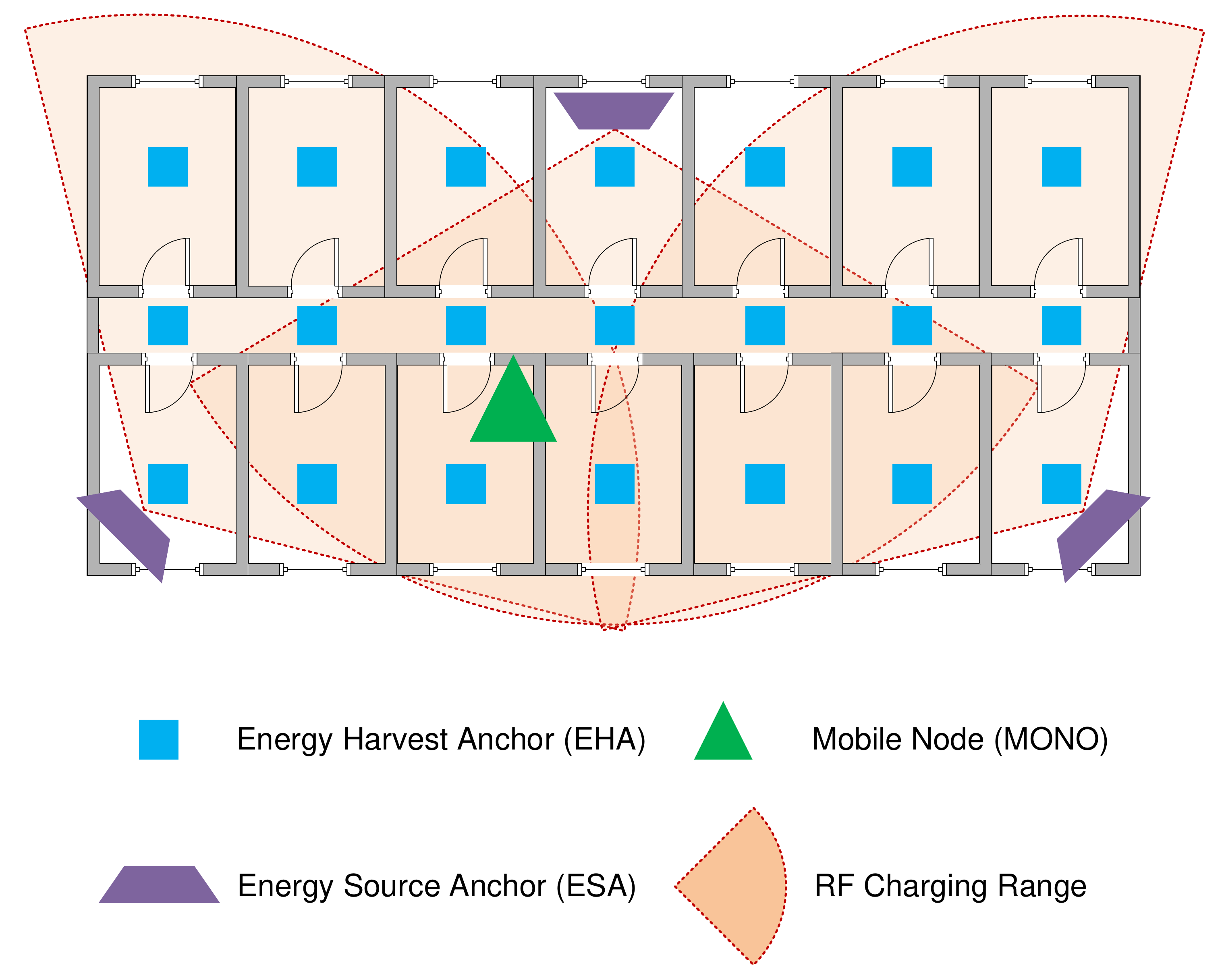}
\caption{An example of indoor BEH deployment with energy harvesting anchor (EHA), energy source anchor (ESA), and mobile node (MONO).} 
\label{pic:system_model}
\end{figure}

We design three types of nodes in {\name}. An example deployment scenario is shown in Fig.\ref{pic:system_model}. 

\begin{itemize}
	\item Energy Source Anchor (\textbf{{\charbe}}):  {\charbe} consists of two components. The first component is called \textbf{{\charbe}-Transmitter}, which is the dedicated RF energy sources using fixed frequency and Tx power. The second component is \textbf{{\charbe}-Controller}, which is a BLE beacon device attached to {\charbe}-Transmitter for controlling the operation (e.g. energy transmitting ON/OFF) of {\charbe}. 
{\charbe}s are statically deployed to emit energy to nearby devices. 
{\charbe} uses fixed power supply via a power cable. Therefore, it is unnecessary to save the communication energy of {\charbe}-Controller. We make {\charbe}-Controller listen to the communication channel in full time. 
	\item Energy Harvesting Anchor (\textbf{{\anbe}}): {\anbe} is a BLE beacon device connected with RF energy harvester.  
We assume {\anbe}s are uniformly deployed in the center of equally partitioned square cells. 
Each {\anbe} is pre-programmed with a unique ID representing the location of the square cell where it resides. 
{\anbe} does not store the harvested RF energy in battery. It directly utilizes the harvested energy for beacon operations. 
	\item Mobile Node (\textbf{{\mono}}): {\mono} has the same hardware structure as {\anbe}. It is attached on the mobile user and powered by RF wireless energy. {\mono} periodically wakes up from sleeping mode and sends beacon request to the neighbor {\anbe}s. We assume that it has a pre-defined table including the position and IDs of {\anbe}s. Once receiving the beacon message of the nearest {\anbe}, the {\mono} first decodes the {\anbe} ID in the message. Then it uses the pre-defined table to correlate the ID with a location, which represents the present proximity position of itself. {\mono} does not have batteries, while directly uses harvested energy from nearby RF energy transmitter for beacon operations. 
\end{itemize}

Based on the three types of nodes, we define a round of \textit{beacon procedure} in {\name} as follows. 

\begin{enumerate}[(i)]
	\item The {\charbe} constantly transmits RF energy to the nearby space. The {\mono}s and {\anbe}s harvest RF energy all the time in the effective RF energy range of {\charbe}.
	\item The {\mono} periodically broadcasts beacon request to the neighbor {\charbe}. 
	\item The {\charbe} that receives the beacon request notifies the {\anbe}s that are nearby the {\mono}. 
	\item The {\anbe}s receive the notification from {\charbe} and send beacon reply back to the {\mono}. 
	\item The {\mono} decodes the content of beacon packet from the nearest {\anbe}, and estimates its own proximity position based on the decoded beacon packet. 
\end{enumerate}

We evaluate the performance of {\name} beacon procedure by the metrics as follows: 

\begin{itemize}
	\item \emph{Packet Reception Rate (PRR)}: Suppose the {\mono} sends beacon request periodically. If the {\mono} successfully receives and decodes the {\anbe} beacon, we record this round of beacon as a success. Otherwise, it is recorded as a failure. PRR is defined as the number of success beacon divided by the number of beacon requests initiated by the {\mono}. PRR represents the success rate to finish beacon communication. 
	\item \emph{Proximity Detection Accuracy (PDA)}: Suppose the {\mono} is inside the cell area of {\anbe} $i$ and sends beacon request periodically. If the {\mono} receives and decodes the beacon message from {\anbe} $i$, we record this round of proximity detection as correct. Otherwise, it is recorded as wrong. PDA is defined as the number of correct proximity detection divided by the number of success beacons. 
PDA represents the correct rate that {\mono} localize itself in the closest {\anbe} cell. 
	\item \emph{Localization Error (LE)}: Suppose the {\mono} is inside the cell area of {\anbe} $i$ and sends beacon request periodically. If the {\mono} receives and decodes the beacon message from {\anbe} $j$, we record the localization error of this round of beacon as the distance between {\anbe} $i$ and $j$. If this round of proximity detection is correct, the error is zero. LE is calculated as the average value of recorded localization errors. LE is used as a complement index of PDA to illustrate the performance of {\name}. 
\end{itemize}

The main aim of {\name} system is to achieve a high PRR. 
Based on this aim, we require PDA has room level accuracy. 
The trade-off relation between PRR and PDA is not the aim of this research. 

\section{Key Components to Achieve Bateryless BLE Beacon Using RF Energy Harvesting}
\label{sec:key_tech}

In this section, we discuss the key components to achieve batteryless BLE beacon using RF energy harvesting. 
These components decrease the power consumption of beacon communication and increase the efficiency of utilizing harvested energy, including reducing time length of receiving messages by collision based beacon (Sec.\ref{sec:collision_beacon}), reducing idle listening time length by passive wakeup (Sec.\ref{sec:passive_wakeup}), reducing unnecessary beacons by range estimation (Sec.\ref{sec:distance_estimation}), and increase energy utilization efficiency by two wave beacons (Sec.\ref{sec:two_wave_wakeup}). 

\subsection{Collision based Beacon}
\label{sec:collision_beacon}

In this section, we focus on decreasing the power consumption of {\mono}. 
We minimize the time length of receiving beacon messages in {\mono}. 
We assume {\anbe}s listens to the communication channel in full time in this section. 
The approach to decrease the idle listening of {\anbe}s is present in the next section (Sec.\ref{sec:passive_wakeup}). 

Most existing communication solutions avoid collisions among packets, such as CSMA/CA and frequency hopping. Although these approaches are quite effective, they involve extra energy consumption for coordinating transmission timing and frequency selection. 
For example, CSMA/CA requires that the device listens to the channel to confirm whether the channel is free. 
However, the harvested RF energy is quite limited. According to our measurement, it is only enough for receiving and transmitting for a very short time. There is not enough energy for coordinating time slot and frequency. 
In addition, when multiple {\anbe}s send beacon messages to a {\mono}, the {\mono} still needs to receive the beacon messages from multiple {\anbe}s one after another. 
However, {\name} requires {\mono} to receive the beacon messages only from the nearest {\anbe}. 
Therefore, {\mono} wastes much energy on receiving unnecessary beacon messages by these existing solutions. 

Contrary to the existing solutions, we take advantage of "packet collision" to decrease the operational time of {\mono} as much as possible.  
We make all {\anbe}s send their packets at the same time, while {\mono} leverages the capture effect~\cite{arnbak:jsac:1987} to decode the strongest signal. 
The key advantage of this method is its energy efficiency, by which {\mono} only needs to be active for a single time slot for packet receiving. 

\subsubsection{Algorithm Design}
\label{sec:algorithm_design}

We present the detail design of collision based beacon as follows, including communication synchronization, orthogonal code, error correction, encoding process, and decoding process. 

\paragraph{Communication Synchronization}
\label{sec:packet_synchronisation}

We make each beacon packet consists of a preamble of one byte, a payload and a CRC of the whole packet. 
The beacon initiative message from {\mono} is named as \texttt{beacon-request}, and the reply beacon message as \texttt{beacon-reply}. 
According to the capture effect, 
the stronger signal can be decoded, if the second beacon arrives while processing the header of the previous beacon. 
Therefore, to leverage the capture effect, the packets from multiple {\anbe}s should arrive at {\mono} within each other's preamble. 
This gives us much space to synchronize the beacons from {\anbe}s. 
The challenges then becomes the synchronization of multiple \texttt{beacon-reply} messages. 

On one hand, if we make {\anbe}s wake up in duty cycle mode, they will not receive the \texttt{beacon-request} packet synchronously. Then the capture effect will not work. 
On the other hand, if the \texttt{beacon-request} packet is used to synchronize the operation of {\anbe}s, {\anbe}s must idle listening the communication channel in full time. This is not applicable due to the limited harvested energy of {\anbe}s. 
Moreover, although there are many existing synchronization approaches, these synchronization process will consume extra energy. 

Therefore, we assume all {\anbe}s are in listening mode in this section. 
We focus on decreasing the power consumption of {\mono}, 
while minimize the power consumption of {\anbe}s in the next section. 
To achieve this synchronization, the {\mono} broadcasts a \texttt{beacon-request} packet. 
The {\anbe}s use this packet as a synchronization signal that instructs all the receiving {\anbe}s to send \texttt{beacon-reply} with the information and their ID encoded in the payload simultaneously. 

\paragraph{Orthogonal Spreading Code}
\label{sec:orthogonal_codes}
Leveraging the capture effect alone however has limitations~\cite{velze:percom:2013}. 
For utilizing capture effect, 
the strongest signal needs a certain minimum signal-to-interference-plus-noise ratio (SINR). 
If this requirement is not satisfied, packets will collide and the content will not be retrieved. 
To overcome this limitation the packet payload is encoded with an orthogonal code to increase inter-packet distinction.  

Assume the value set size of beacon packet payload is limited, 
and a unique orthogonal code corresponds to a element in the value set. 
Take the {\anbe} ID in the payload as an example. 
Each {\anbe} ID corresponds a unique orthogonal code. 
In the encoding process at {\anbe}, 
each bit of the {\anbe} ID is multiplied by an orthogonal code unique for this {\anbe}. 
The encoded {\anbe} ID is then send in the payload of a packet.
The decoding process at {\mono} is an XOR operation between the payload of received packet with a list of orthogonal codes. Any method can be used as long as the codes all have zero cross-correlation with each other. In this paper a Hadamard matrix of size $k$ is used to generate the codes, i.e.,
\begin{equation}
H_{2^k} = \begin{bmatrix} H_{2^{k-1}} & H_{2^{k-1}}\\ H_{2^{k-1}} & -H_{2^{k-1}}\end{bmatrix} = H_2\otimes H_{2^{k-1}},
\end{equation}
where $H_2 = \begin{bmatrix} 1 & 1 \\ 1 & -1 \end{bmatrix}$, $2 \leq k \in \mathbb{N}$
and $\otimes$ is the Kronecker product.

In our implementation, 
we define an orthogonal code for each {\anbe} ID. 
For simplifying the implementation, 
we pre-install all the orthogonal codes in {\mono} before the experiment. 
The value set of beacon packet payload is all the {\anbe} IDs. 
The size of the orthogonal code is 16 bits and the total message length is 30 bytes in the payload of the beacon packet. 

\paragraph{Forward Error Correcting Code}

To further increase the success rate for decoding packet, the {\anbe} ID is first encoded with a \ac{FEC} code, before it is multiplied by the orthogonal codes. The \ac{FEC} is constructed by maximum minimum Hamming distance codes~\cite{macdonald1960design} i.e. codes having equal Hamming distance to each other. 
The decoding of FEC uses minimum distance decoding.

\paragraph{Encoding Process}
All the {\anbe}s have an array of FEC codes with equal Hamming distance to each other. 
Every {\anbe} ID is FEC encoded by replacing the ID with an unique FEC code from the array. 
The orthogonal code array is generated using the method described in Sec.~\ref{sec:orthogonal_codes}.
For the generated matrix, each $-1$ symbol is replaced with a $0$ and each row of the matrix represents a binary spreading code. 
Finally, every bit of the FEC code is represented by an unique orthogonal code in the orthogonal code array. 
If the bit is zero, the bitwise NOT of the orthogonal code is used. 

\paragraph{Decoding Process}
We assume the {\mono} has an array of all the orthogonal and \ac{FEC} codes used by the {\anbe}s.
For every entry in the orthogonal code array, the {\mono} tries to decode the packet. 
Suppose the {\anbe} ID is encoded inside packet payload. 
When a candidate {\anbe} ID is found, the decoded code is compared to the correlating FEC code of the candidate ID. 
Suppose each FEC code has equal Hamming distance $d$ to each other. 
We compare the Hamming distance $d_c$ of the candidate FEC code to a code from the FEC array. If $d_c < \frac{d}{2}$, we assume that the candidate ID is the correct one.

\subsubsection{Validity Test}
\label{sec:validity_testing}

\begin{figure}
	\centering
	\subfigure[Without orthogonal codes]{
		\includegraphics[width=0.8\columnwidth]{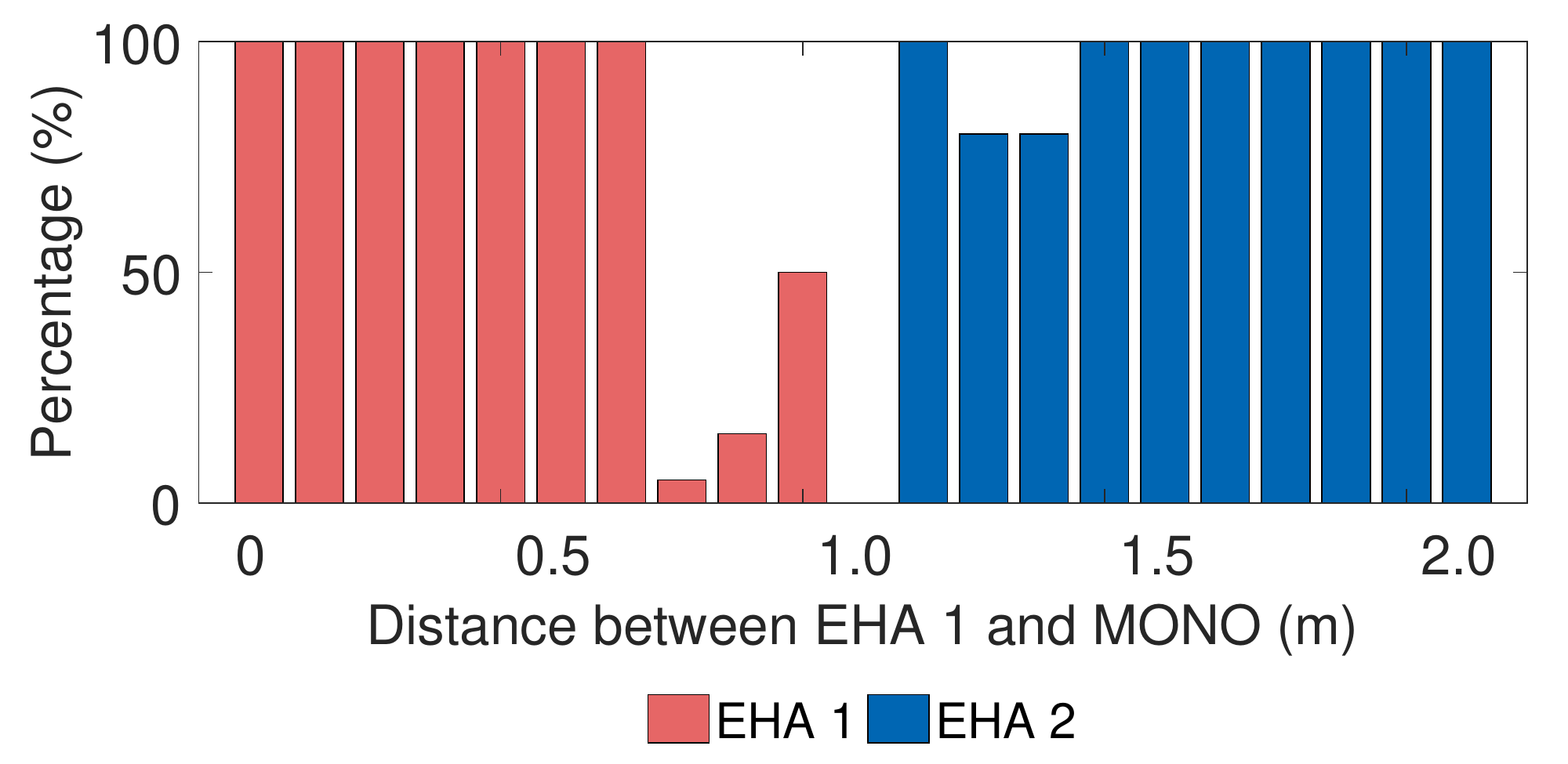}
		\label{exp:localization_before}
	}
	\subfigure[With orthogonal codes]{
		\includegraphics[width=0.8\columnwidth]{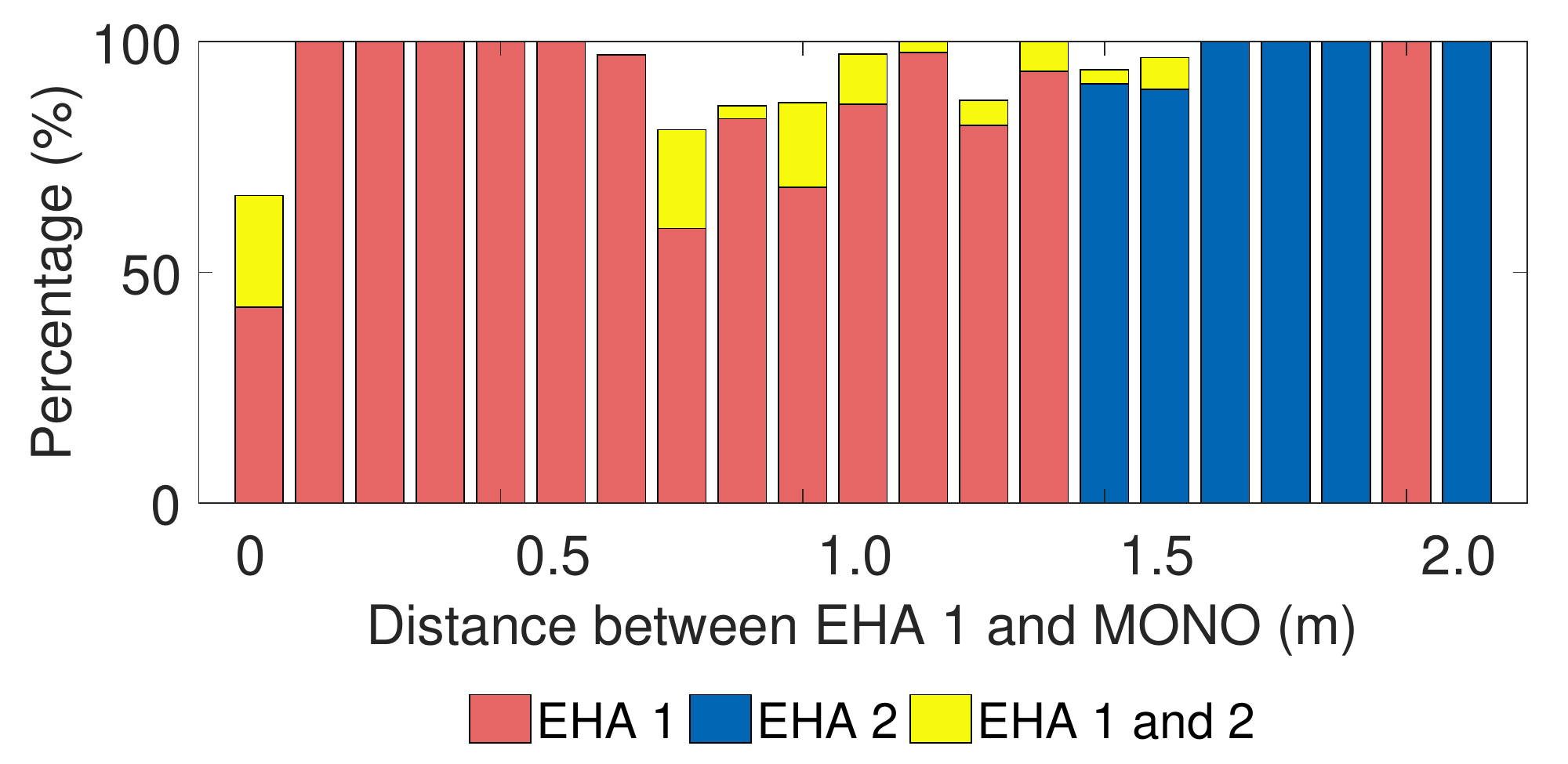}
		\label{exp:localization_after}
	}
	\caption{In the experiment of collision based beacon, two static {\anbe}s are placed two meters apart. We mark {\anbe} 1 at 0m and {\anbe} 2 at 2m in the x coordination of the figures. Compared with the results of (a), the results of (b) eliminate the packet reception ``dead zone'' using orthogonal codes.}
	\label{fig:localization_before_after}
\end{figure}

To verify that the collision based beacon works correctly, the following experiment was performed. 
We select the Smart BLE Beacon Kit form Nordic Semiconductor~\cite{nrf51822} as the platform of {\anbe}. This module has a coin size form with a PCB integrated antenna. 
To increase the radio sensitivity of receiver, we select the Nordic Semicoductor PCA10005 with BLE~\cite{nrf51822} as the platform of {\mono}. Compared with {\anbe}, {\mono} has an SMA connector with a connected quarter-wave helical monopole antenna of 1.6\,dBi gain. 
These two platforms are both the nRF51822 SoC with a ARM Cortex M0 from Nordic Semiconductor with BLE~\cite{nrf51822}. 

Two {\anbe}s are placed two meters apart and a {\mono} is placed at 20 measuring points in a straight line between the {\anbe}s. 
The interval distance of the measuring points is 10\,cm. 
The transmission power of the {\anbe}s is set to 0\,dBm. The {\mono} stays at each measuring point for 30\,s and sends \texttt{beacon-request} every second. 
The two {\anbe}s send their \texttt{beacon-reply} back immediately after receiving the \texttt{beacon-request}. 

The experiment is first performed without orthogonal codes, and the results are shown in Fig.~\ref{exp:localization_before}. 
After that, the same experiment is performed with orthogonal codes, and the results are shown in Fig.~\ref{exp:localization_after}. 
In Fig.~\ref{exp:localization_before}, most of the collision packets can be correctly decoded in capture effect. 
But there is a ``dead zone'' between the two {\anbe}s, i.e. an area of no packet reception. This is because the packets collide but the SINR is insufficient to receive a correct packet in this region. Fig.~\ref{exp:localization_after} shows that the ``dead zone'' is eliminated using orthogonal codes. Moreover, there is some chance that the {\mono} decodes the two {\anbe} IDs from one collision beacon. 

\subsection{Passive Wakeup}
\label{sec:passive_wakeup}

The collision based beacon introduced in Sec.~\ref{sec:collision_beacon} requests synchronization of packet transmission among {\anbe}s. 
However, the harvested energy is not enough for {\anbe}s to continuously listen to the synchronization signal, i.e. \texttt{beacon-request}. 
Therefore, we propose passive wakeup approach to wakeup {\anbe}s from sleeping mode only when the \texttt{beacon-request} is sent from {\mono}. 

\subsubsection{Basis of Passive Wakeup} 

The existing BLE beacon devices do not have passive wakeup function. 
We achieve RF-based passive wakeup by taking advantage of the properties of BLE device, RF energy transmitter and harvester as follows. 

\begin{enumerate}[(i)]
	\item The harvested RF energy of {\anbe} is very sensitive to the variation of the Tx power from RF energy transmitter. {\anbe} can detect the ON/OFF state of {\charbe}-Transmitter by measuring the harvested power. Therefore, changing the Tx power from ON to OFF state in the {\charbe}-Transmitter is used as the wakeup signal to {\anbe}s. 
	\item 
Analog to Digital Converter (ADC) is widely used to measure the voltage level of batteries. 
ADC of Nordic nRF51822 enables sampling of external signals through a front-end multiplexer. 
Using the same BLE device as shown in Sec.\ref{sec:validity_testing} (Smart Beacon Kit form Nordic Semiconductor with BLE~\cite{nrf51822}), we measure the power consumption of different operations as shown in Table~\ref{tab:power_consumption_wpa}. 
The power consumption of ADC measurement is much lower than transmitting and receiving packets. Moreover, it is of the same magnitude as the harvested power (from $\approx$\,3.2\,mW at 1\,m to $\approx$\,0.79\,mW at 3\,m). Although the power consumption of ADC measurement is larger than harvested power at the distance of 3\,m, {\anbe} can conduct ADC measurement periodically and be in sleeping mode the rest of time. This ensures that the average power consumption of periodical ADC measurement is lower than the harvested energy. If the {\charbe} is switched off shortly for sending passive wakeup signal, the remaining energy in the capacitor of {\anbe} is enough for its ADC operation. Therefore, we make {\anbe}s periodically wakeup from sleeping mode and perform the measurement of the the harvested energy voltage on ADC port. 
	\item {\charbe}s are deployed with fixed power supply for transmitting wireless energy. Therefore, {\charbe}-Controller, which attaches on {\charbe}-Transmitter, has enough energy to listen to the communication channel continuously. We make {\charbe}-Controller listens to the \texttt{beacon-request} signal from {\mono} in full time. 
\end{enumerate}

\begin{table}[t]
	\caption{Measurements about Energy consumption of {\anbe} node operations.}	
	\center
	\begin{threeparttable}
	\begin{tabular}{ c | c | c | c }
	    \hline
		\textbf{Operation} & {\textbf{Power (mW)}} & {\textbf{Time (ms)}} & {\textbf{Energy ($\mu$J)}}\\
		\hline
		\hline
		Transmitting& 35.88 & 0.80 & 28.7\\
		\hline
		Receiving & 20.17 & 0.60 & 12.1\\
		\hline
		ADC & 1.69 & 0.65 & 1.10\\
		\hline
		Sleeping & 0.14 & $-$ & $-$\\
		\hline
	\end{tabular}
	\begin{tablenotes}
		\scriptsize \item $\S$ Voltage supply is 3 V. The payload length of BLE packet is 30 bytes. 
	\end{tablenotes}
	\label{tab:power_consumption_wpa}	
	\end{threeparttable}
\end{table}

\subsubsection{Procedure of Passive Wakeup} 
\label{sec:workflow_passive_wakeup}

\begin{figure}[t]
\centering
\includegraphics[width=0.8\columnwidth]{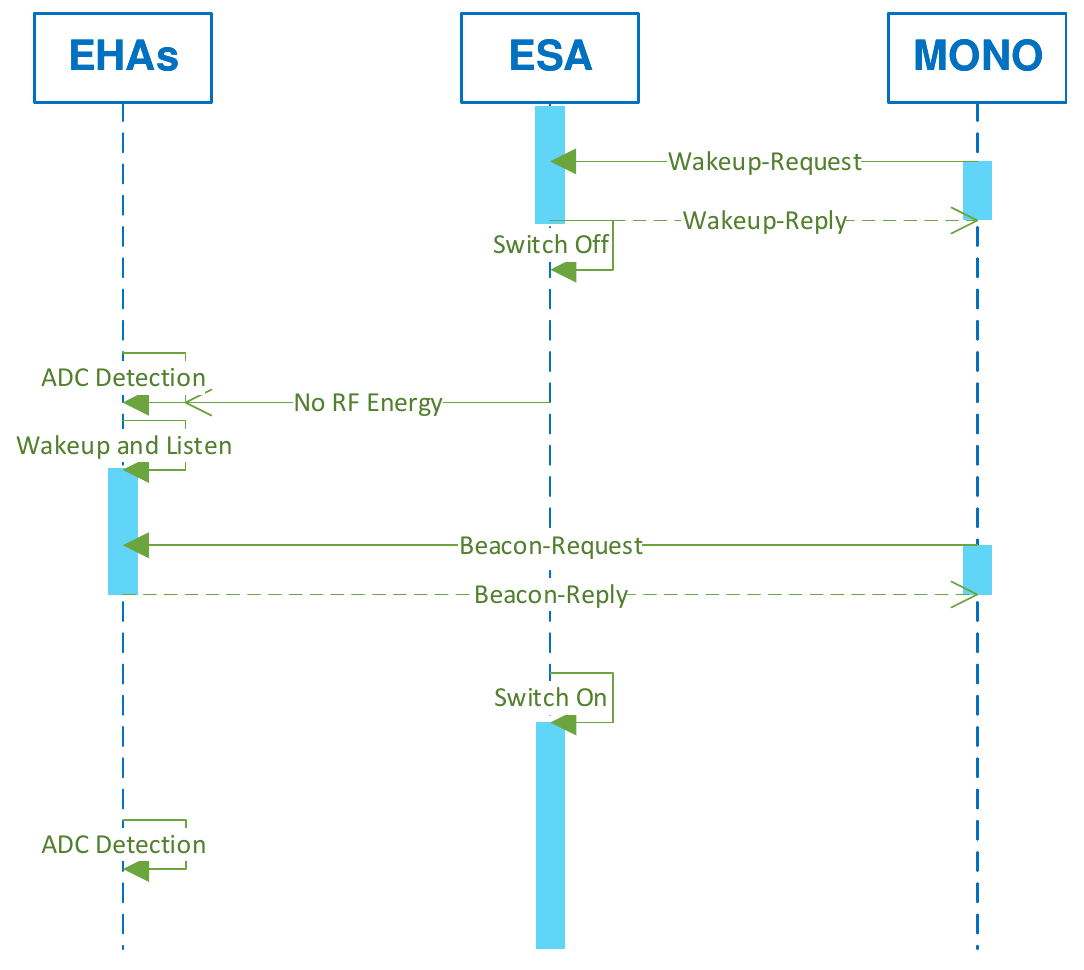}
\caption{Workflow of passive wakeup process based on RF energy transmitting.} 
\label{pic:passive_wakeup}
\end{figure}

Now, we combine the above three properties together to implement passive wakeup.

First of all, each {\anbe} makes ADC detection with period $t_c$. {\mono} broadcasts a \texttt{wakeup-request} message to {\charbe}. 
The {\charbe}-Controller sends a \texttt{wakeup-reply} message back immediately after receiving \texttt{wakeup-request}. 
This \texttt{wakeup-reply} from the {\charbe} is used to notify the {\mono} that {\anbe}s will be wakeup in listening mode. 
After sending the \texttt{wakeup-reply}, the {\charbe} sends a passive wake up signal to all the neighbor {\anbe}s by switching off the Tx power of {\charbe}-Transmitter for a short time and then switch on again. 
Once a voltage falling of the harvested power is measured using ADC detection by a {\anbe}, it wakes up from sleeping mode and starts listening. 
Meanwhile, the {\mono} receives the \texttt{wakeup-reply} from the {\charbe} and waits for $t_c$ to broadcast the \texttt{beacon-request} to the {\anbe}s. 
After receiving the \texttt{beacon-request} from the {\mono}, the {\anbe}s send \texttt{beacon-reply} back to the {\mono}. 
Then the {\mono} decodes the packet by the collision based beacon approach as explained in Sec.~\ref{sec:collision_beacon}. 
In this workflow, each {\anbe} periodically wakes up and conducts ADC measurement to detect passive wakeup signal. 
To guarantee that all the {\anbe}s can hear the \texttt{beacon-request} signal from the {\mono}, each {\anbe} wakes up immediately after receiving the passive wakeup signal and listens a maximal period of $t_c$. 
If a {\anbe} receives nothing in that time frame, then it sleeps again.  
The whole process is illustrated in Fig.~\ref{pic:passive_wakeup}. 

\subsubsection{Optimization on Passive Wakeup}
\label{sec:opt_passive_wakeup}

To further reduce the average power consumption of {\anbe}, we optimize the ADC detection period $t_{{c}}$ of {\anbe}. 

Suppose {\mono} requests beacon from {\anbe}s with period $t_m$. 
Name the average power consumption of a {\anbe} in $t_m$ as $P_a$. 
As explained in Sec.\ref{sec:workflow_passive_wakeup}, {\charbe}-Transmitter switches OFF as wakeup signal after sending \texttt{wakeup-reply}. 
Suppose the time length from switching {\charbe}-Transmitter OFF until detecting passive wakeup signal by the {\anbe} is $t_{u}$. 
After receiving \texttt{wakeup-reply} from {\charbe}, 
the {\mono} sleeps for $t_{c}$ and then sends \texttt{beacon-request}. 
Then the {\anbe} is waken up and stays in receiving state for time length $t_{{rx}} = {t_c} - {t_u}$. 
Denote ${k_{{d}}} = \left\lfloor {{{ \frac {{t_m}}{{t_c}}}}} \right\rfloor$ as the number of ADC measurement during $t_m$. 
Name $t_{{d}}$, $t_{{rx}}$ and $t_{{tx}}$ as the time spent on ADC measurement, receiving packet, and packet transmission by {\anbe}, respectively. 
Name $P_{{d}}$, $P_{{rx}}$, $P_{{tx}}$, and $P_{{s}}$ as the power consumption of ADC measurement, packet reception, packet transmission and sleeping by {\anbe}, respectively. 
We assume the time $t_{u}$ of detecting passive wakeup signal happens uniformly in the ADC measurement period $t_{c}$. 

\begin{figure}[t]
\centering
\subfigure[Setup map.]{\includegraphics[height=0.35\columnwidth]{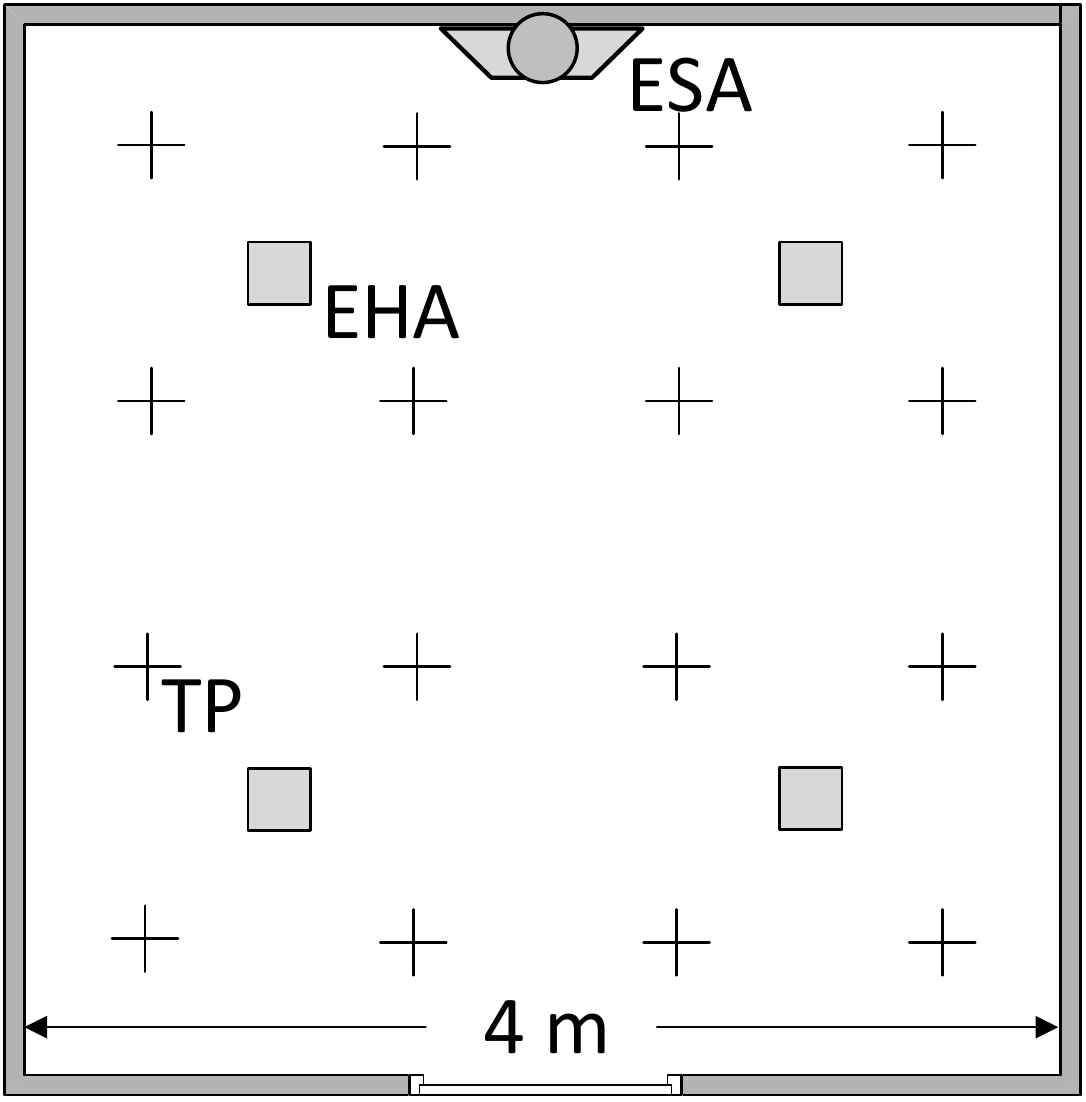}\label{fig:collision_drop_setup}}
\subfigure[PRR and PDA results.]{\includegraphics[height=0.55\columnwidth]{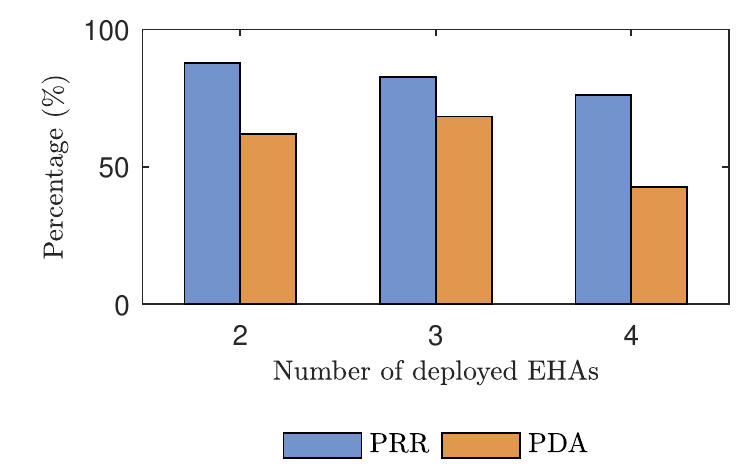}\label{pic:cell_level_performance_drop}}
\subfigure[Error distance of LE.]{\includegraphics[height=0.5\columnwidth]{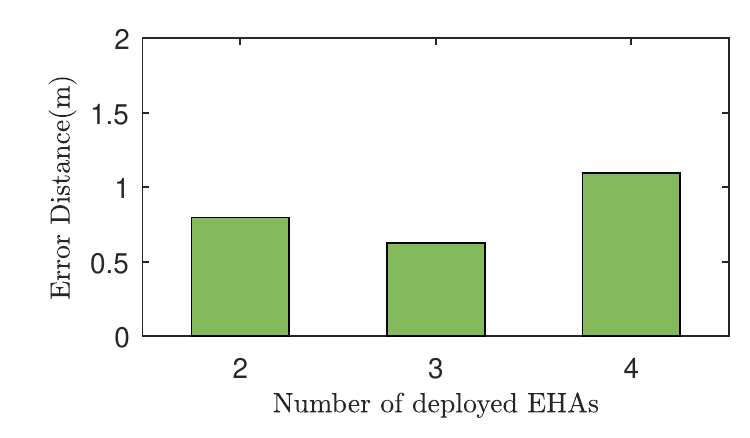}\label{pic:cell_level_performance_drop_ed}}
\caption{(a) The places marked as "$+$" are the testing positions (TP) of {\mono}. (b) Packet reception rate (PRR) and proximity detection accuracy (PDA) of collision based beacons in different number of {\anbe} nodes deployed. The radio interference from multiple beacon nodes causes the decrease of PRR and PDA. (c) Error distance of LE in proximity detection.}
\label{fig:rss_distance}
\end{figure}

\begin{proposition}

The value of $t_{c}$ producing minimum expected power consumption ${{\mathrm{E}(P_a)}}$ at {\anbe} is 

\begin{equation}
\arg \min_{t_c}{\left[ {{\mathrm{E}(P_a)}} \right]} \approx \sqrt {{{\frac {2{t_{m}}{t_{{d}}}{P_{{d}}}} {{P_{{rx}}}}}}}.
\label{eq:min_value_para}
\end{equation}
\end{proposition}

\begin{proof}
The CDF of $t_u$ is $F(t_u) = {{\frac {t_u} {{t_c}}}},t_u \in [0,{t_c})$, and the expectation of $t_u$ is $\mathrm{E}(t_u) = {{\frac {{t_c}} 2}}$. 
Then the expected receiving state time is $t_{{rx}} = t_{c} - \mathrm{E}(t_u) = {{\frac {{t_c}} 2}}$. 
The time of sleeping during $t_{m}$ is $t_{{s}} =  t_m - (k_{{d}}t_{{d}} + t_{{rx}} + t_{{tx}})$. 
Then the average power consumption of a {\anbe} in a beacon period $t_{m}$ is 

\begin{equation}
\begin{split}
P_{{a}} &= \frac{P_{{rx}}t_{{rx}}+P_{{tx}}t_{{tx}}+k_{{d}}P_{{d}}t_{{d}}+P_{{s}}t_{{s}}}{t_m} \\
&= \frac{{{P_{{{rx}}}}}}{{2{t_m}}}{t_c} + {P_{{d}}}{t_{{d}}}\frac{1}{{{t_c}}} + H, 
\end{split}
\label{eq:power_consumption}
\end{equation}

where $H = {{\frac {{P_{{{tx}}}}{t_{{{tx}}}} + {P_s}{t_s}} {{t_m}}}}$. 
As $k_{{d}}$ is discrete value, we replace it with a continuous value ${k^{c}_{{d}}} = {{\frac {{t_m}} {{t_c}}}}$ for estimating the minimum value $\mathrm{E}(P_{a})$. 
The value of $t_{c}$ that minimizes the expectation of $P_{a}$ is calculated as 
$\arg\min_{t_c}{\left[ {\mathrm{E}({P_a})} \right]} = \{ {t_c}\left| {{{\frac {\partial {P_a}} {\partial {t_c}}}} = 0,\ {k_{{d}}} = k_{{d}}^{{c}}} \right.\}$ which results in (\ref{eq:min_value_para}). 
\end{proof}

\subsection{Range Estimation}
\label{sec:distance_estimation}

The passive wakeup approach reduces the power consumption of idle listening in {\anbe}s as present in Sec.\ref{sec:passive_wakeup}. 
However, there are two problems as follows. 

\begin{enumerate}[(i)]
	\item The passive wakeup approach activates all the {\anbe}s near the {\charbe}. 
As the number of {\anbe}s that send collision based beacon increases, 
the decoding error in {\mono} might increase. 
To prove this assumption, we deploy 2, 3 and 4 {\anbe}s respectively around a {\charbe}, and test the PRR, PDA and LE. 
The deployment setup and testing positions are illustrated in Fig.~\ref{fig:collision_drop_setup}. 
The test results are shown in Fig.~\ref{pic:cell_level_performance_drop} and Fig.~\ref{pic:cell_level_performance_drop_ed}. 
The key result is that the PRR and PDA decrease as the number of {\anbe}s increases. 
This is mainly caused by the radio interference to the orthogonal code from multiple {\anbe}s. 
It means that radio interference caused by multiple {\anbe}s limits the performance and scalability of {\name} system. 
	\item The aim of {\name} is to make {\mono} receive information, such as node ID, from the nearest {\anbe}. It is more energy efficient to wakeup only the {\anbe}s which are close to the requesting {\mono}. 
\end{enumerate}

\begin{figure}[t]
\centering
\subfigure[RSS of BLE beacon.]{\includegraphics[height=0.4\columnwidth]{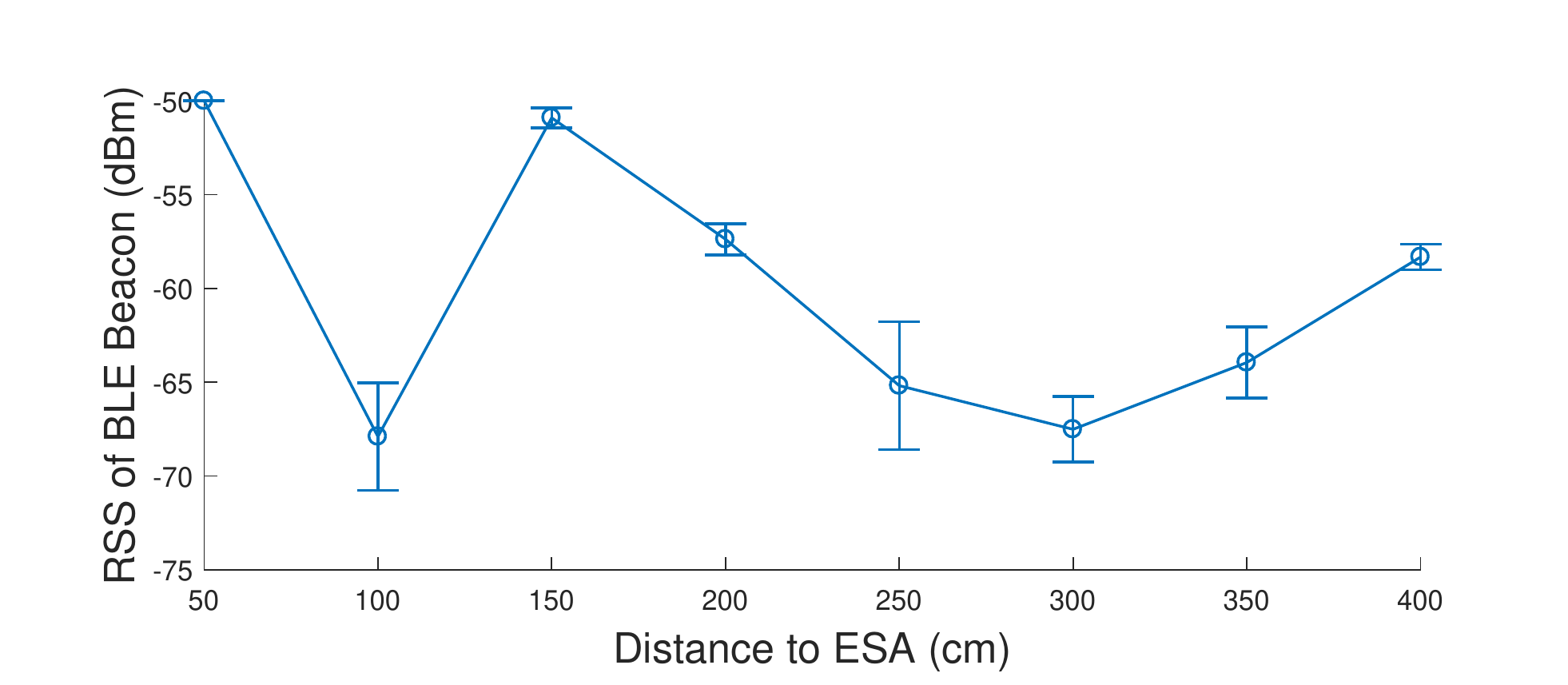}\label{fig:ble_rssi}}
\subfigure[Received power in RF energy harvester.]{\includegraphics[height=0.4\columnwidth]{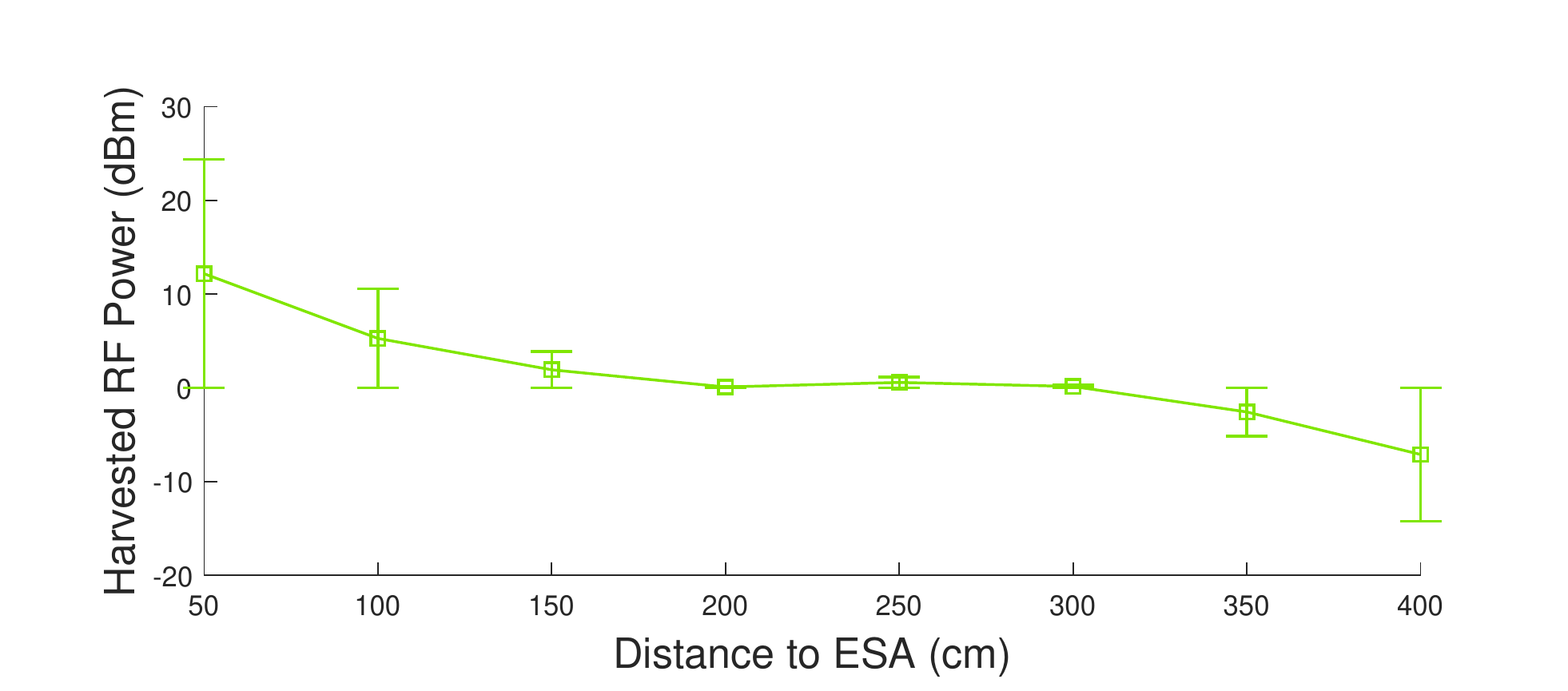}\label{fig:harvested_power}}
\caption{(a) The RSS of BLE beacon in the nRF51822 device. (b) The harvested power in the RF energy harvester Powercaster P2110.}
\label{fig:rss_eh_distance}
\end{figure}

Based on the analysis above, 
we need to restrict the number of {\anbe}s used for beaconing. 
We propose to use the harvested RF energy of {\mono} for estimating its possible range. This range information is used to select the {\anbe}s that are near {\mono} for sending collision based beacon. 

\subsubsection{Properties of Harvested Power} 

We take advantage of the power of harvested RF energy to make range estimation for two reasons.

\begin{enumerate}[(i)]
	\item The power of harvested RF energy is more stable than the received signal strength (RSS) of BLE packets at the same distance. As a comparison, we measure the power of harvested RF energy using ADC and the RSS of BLE packets at various distances. At each testing position, we record 100 measurement results. The measurement results are shown in Fig.~\ref{fig:rss_eh_distance}.  The RSS of BLE fluctuates much, which is difficult to be used for distance estimation. The power of harvested RF energy is rather stable. Therefore, we take advantage of the attenuation of harvested power over distances to estimate the possible location range of {\mono}. 
	\item {\charbe} constantly transmits RF energy to the nearby space. Its transmission radio covers the whole deployment area of {\anbe}s. We do not need to deploy any additional communication component for range estimation. 
\end{enumerate}

\subsubsection{Procedure of Range Estimation}

Suppose {\anbe}s $\{1,...,i,...,n\}$ are deployed uniformly around the {\charbe}. 
The ADC based power measurement results of the {\mono} and {\anbe} are $\xi$ and $\mu$, respectively. 
At the beginning of the range estimation process, 
each {\anbe} reports its power measurement $\mu$ to the {\charbe}. 
Name the measured power of {\anbe} $i$ as $\mu_{i}$. 

Although the power of harvested RF energy is more stable than the RSS of BLE beacon packet, 
it is inaccurate to use it for calculating the precise location of {\mono}. 
As explained in Sec.\ref{sec:system_model}, {\charbe} is pre-programmed with the position map of {\anbe}s. 
Therefore, we define a threshold value $\rho$ of $\mu$ to categorize the {\anbe}s into two parts: the {\anbe}s which are closer to the {\charbe}, and the {\anbe}s which are farther to the {\charbe}. 
To set the value of $\rho$, 
we pick the {\anbe} which is at the geographic middle position of the deployed {\anbe}s. Then we set $\rho$ equal to the measured harvested power of the picked {\anbe}. 

After comparing the value between $\mu$ and $\rho$, 
we name the two sets of {\anbe}s as $\Delta_{c} = \{ ${\anbe}$\  i |\ \mu_{i}  \ge \rho \}$ and $\Delta_{f} = \{ ${\anbe}$\ {i}\ |\ \mu_{i}  < \rho \}$. 
To request beacon services, 
the {\mono} broadcasts a \texttt{beacon-request} with $\xi$. 
After receiving the \texttt{beacon-request}, the {\charbe} compares $\xi$ with the threshold value $\rho$. 
If $\xi \leq \rho$, the {\charbe} sends the passive wakeup signal and broadcasts a \texttt{sleep} packet with the IDs of the {\anbe}s in $\Delta_{c}$. 
If $\xi > \rho$, the {\charbe} sends the passive wakeup signal and broadcasts a \texttt{sleep} packet with the IDs of the {\anbe}s in $\Delta_{f}$. 
Once woken up, the {\anbe}s listen to the message from the {\charbe}. 
If the {\anbe} does not receive a \texttt{sleep} command with its ID, it keeps listening to a \texttt{beacon-request} from the {\mono}, otherwise goes to sleep. 
The detailed process is presented in Fig.~\ref{fig:passive_wakeup_id_range}. 

The accuracy of range estimation using harvested RF power is low, which could cause error in range estimation.  
Suppose the range estimation of the {\mono} is wrong, and the BLE radio transmission range of {\anbe} is set to be only inside its cell area. 
Then the {\mono} can not receive a \texttt{beacon-reply} from the {\anbe}s. 
We cope with this error situation as follows. 
If the {\mono} does not receive any \texttt{beacon-reply} in a beacon period, the {\mono} will request the {\charbe} to re-send wakeup signal by \texttt{wakeup-request} packet with ``\textit{request-again}'' in the payload. 
If the {\charbe} sends wakeup signal to the {\anbe}s of $\Delta_{c}$ previously, it sends wakeup signal to the {\anbe}s in $\Delta_{f}$ now. 
If the {\charbe} sends wakeup signal to the {\anbe}s of $\Delta_{f}$ previously, it sends wakeup signal to the {\anbe}s in $\Delta_{c}$ now. 
For the convenience of expression, 
we name the range of {\anbe}s in the first round of wakeup as $\Delta$, and the second round as $\overline \Delta$. 

\begin{figure}[t]
\centering
\includegraphics[width=0.98\columnwidth]{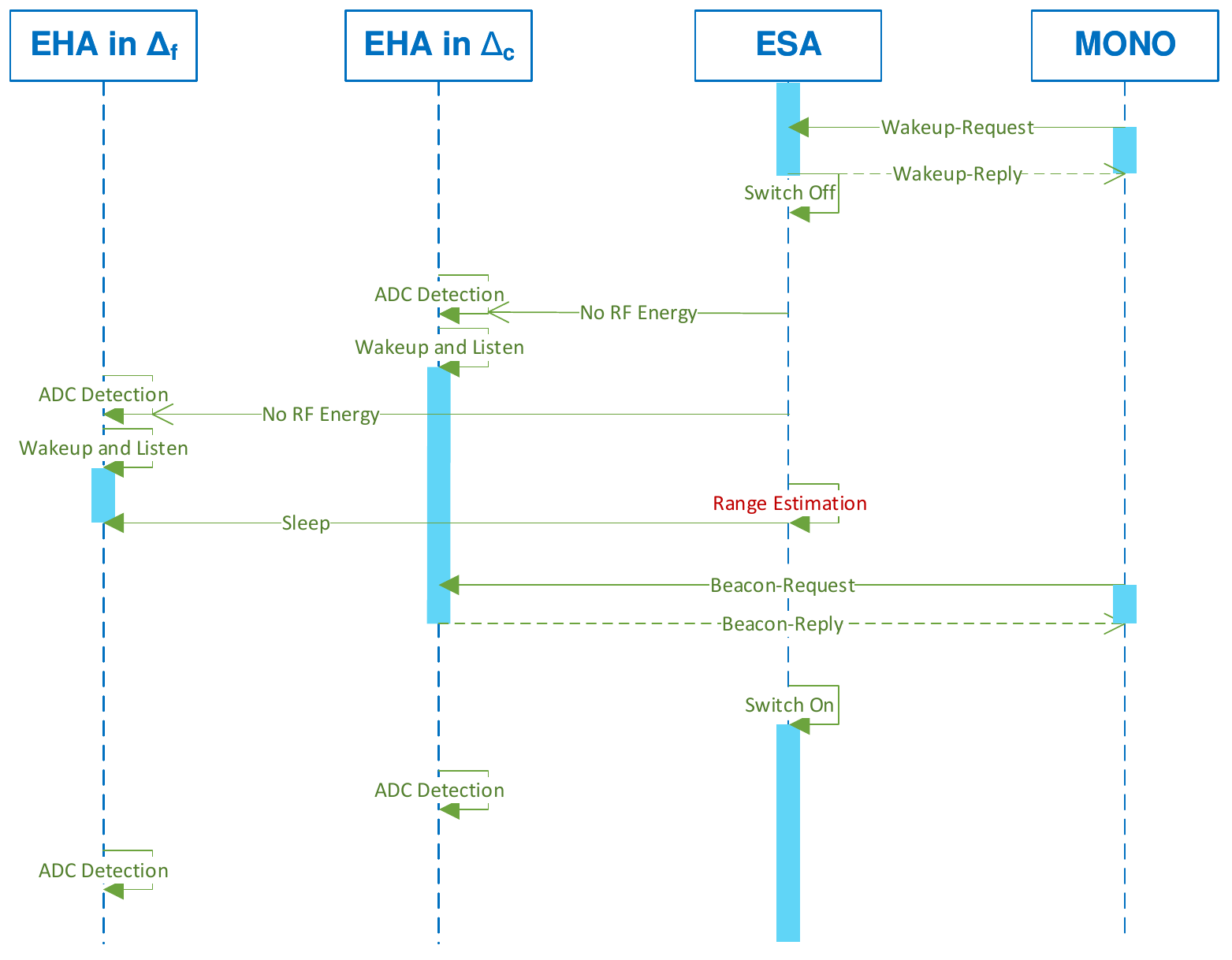}
\caption{{\anbe}s send beacon messages to {\mono} with range estimation. In this example workflow, {\charbe} sends \texttt{sleep} packet with the {\anbe} IDs belonging to $\Delta_{f}$. After receiving the \texttt{sleep} command, {\anbe}s in $\Delta_{f}$ stops listening and changes to sleep mode, and {\anbe}s in $\Delta_{c}$ keep listening to the \texttt{beacon-request}.}
\label{fig:passive_wakeup_id_range}
\end{figure}

\subsection{Two Wave Beacons}
\label{sec:two_wave_wakeup}

The range estimation approach of Sec.\ref{sec:distance_estimation} reduces some unnecessary beacon communication by waking up the {\anbe}s of specified area. 
However, this approach relies on the measurement of harvested power. 
The measurement value of $\xi$ and $\mu$ can be affected by environmental factors, including obstacles between devices, relative direction of antennas between energy transmitter and harvester, destructive RF interference from nearby {\charbe}s, etc. 
Therefore, we propose an approach called \emph{two wave beacons} to reduce part of unnecessary beacons and increase energy utilization efficiency without extra measurements or range estimation. 
The \emph{two wave beacons} approach is based on the two properties of RF energy harvesting as follows. 

\begin{enumerate}[(i)]
	\item According to the measurement results as shown in Fig.\ref{pre:energy_harvest}, the harvested energy by {\anbe}s is quite unbalanced. 
{\anbe} at 1.0\,m to the {\charbe}-Transmitter harvests 36 times more energy than at 3.5\,m. 
Therefore, it is more important to reduce unnecessary beacons and save energy in the {\anbe}s which are further away from {\charbe}. 
	\item According to the collision based beacon in Sec.\ref{sec:collision_beacon}, all the {\anbe}s near the {\mono} must wake up at the same time. 
This means that, in every round of beacon communication, all the {\anbe}s must wait until the {\anbe} which harvests the lowest power has enough energy. 
Therefore, it is necessary to increase the energy efficiency of {\anbe}s which are further away from the {\charbe} by shifting some of their workload to the closer {\anbe}s. 
\end{enumerate}

\subsubsection{Procedure of Two Wave Beacons}
\label{sec:two_wave_procedure}

\begin{figure}
\centering
\includegraphics[width=0.75\columnwidth]{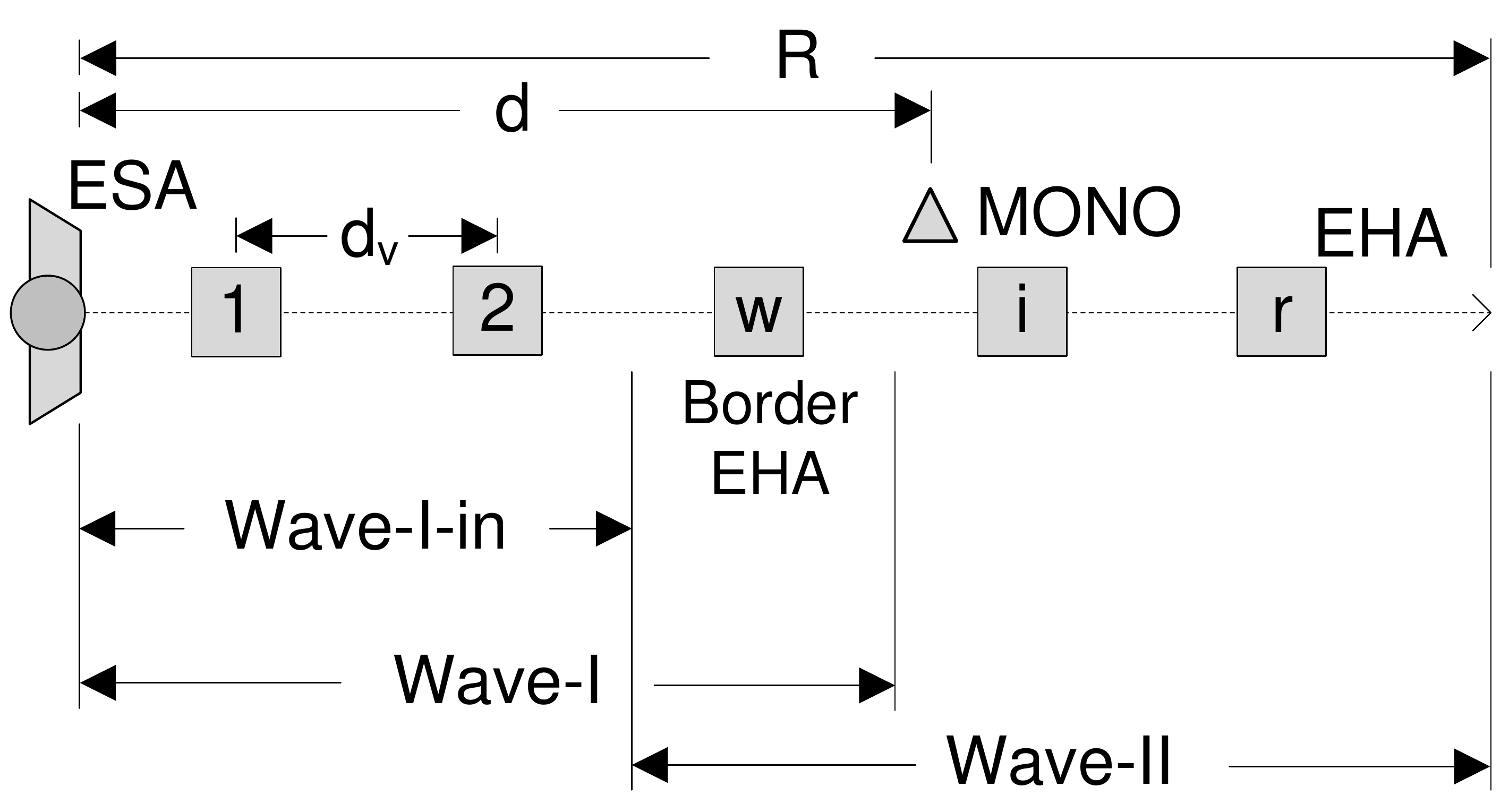}
\caption{Example deployment of two wave beacons. {\anbe} $w$ is selected as the border {\anbe}.}
\label{fig:analysis_setup_twowave}
\end{figure}

The \emph{two wave beacons} approach includes two rounds, i.e. waves, of beacon communication between {\mono}, {\charbe} and {\anbe}s. 
Each wave is similar to the the workflow of Sec.\ref{sec:distance_estimation}. 
The key difference is which {\anbe}s are selected to receive \texttt{sleep} signal. 

We use Fig.\ref{fig:analysis_setup_twowave} as an example of the system setup. 
The same as Sec.\ref{sec:distance_estimation}, {\anbe} measures the mean harvested power and sends this information to the {\charbe}. 
The {\charbe} selects the harvested power from one of the {\anbe}s as the power threshold $\theta$. 
Suppose the harvested power of {\anbe} $w$ is selected as the threshold value $\theta$. 
We call the {\anbe} $w$ as \emph{border {\anbe}}. 
We categorize the {\anbe}s into two groups by the border {\anbe}. 
The first group is the {\anbe}s with $\mu \ge \theta$, which is called Wave-I. 
The second group is the {\anbe}s with $\mu \le \theta$, which is called Wave-II. 
To simplify the explanation, we name the area of Wave-I except the cell of border {\anbe} as Wave-I-in. 
The procedure of \emph{two wave beacons} is as follow, and its sequential diagram is shown in Fig.~\ref{fig:analysis_setup_twowave_seq_diagram}. 

\begin{enumerate}[(i)]
\item \textbf{Step I $-$ First Wave Beacon}: 
{\charbe} wakes up {\anbe}s in Wave-I to send beacon messages to {\mono}. In this step, there are two possible conditions. 
\begin{itemize}
\item 
Suppose the {\mono} is in Wave-I-in. 
The {\mono} will receive the \texttt{beacon-reply} from one of the {\anbe}s in Wave-I-in. 
This beacon packet is decoded for its information and counted as the proximity detection result. 
In the example of Fig.~\ref{fig:analysis_setup_twowave}, the {\mono} will receive the beacon from {\anbe} 1 or {\anbe} 2. 
After that, the procedure of the \emph{two wave beacons} approach finishes. 
\item Suppose the {\mono} is in Wave-II. 
The border {\anbe} always has the closest distance to the {\mono}. 
According to the measurement results of collision based beacon in Sec.\ref{sec:collision_beacon}, 
{\mono} receives the beacon message from the nearest {\anbe}. 
So that, the {\mono} will always decode the beacon packet from the border {\anbe}. 
In this situation, we cannot confirm the position of the {\mono}. 
The {\mono} initiates the second wave by re-sending \texttt{wakeup-request} packet with ``\textit{request-again}'' in the payload. 
\end{itemize}

\item \textbf{Step II $-$ Second Wave Beacon}: 
After the operations of the first wave beacon, we determine that the {\mono} must be inside the Wave-II, although its position is unclear. 
In the second wave beacon, {\charbe} wakes up {\anbe}s in Wave-II to send beacon messages to {\mono}. 
Then the {\mono} will receive and decode the \texttt{beacon-reply} from one of the {\anbe}s in Wave-II. 
This beacon packet is decoded for its information and counted as the proximity detection result. 
The procedure of the two wave beacons finishes.
\end{enumerate}

\begin{figure}
\centering
\includegraphics[width=0.98\columnwidth]{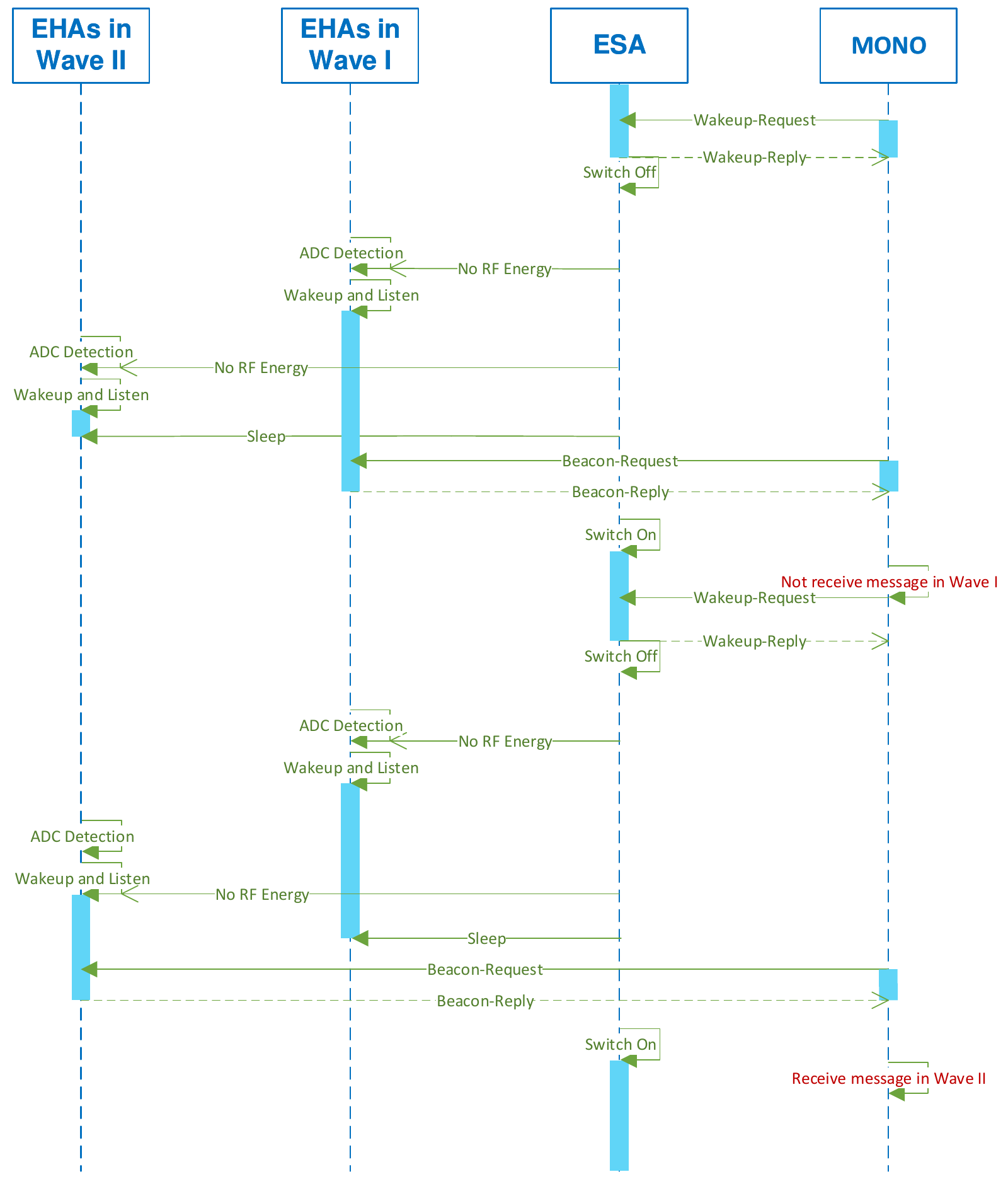}
\caption{Sequential diagram of two wave beacons. In this diagram, we assume the {\mono} is inside Wave-II area. }
\label{fig:analysis_setup_twowave_seq_diagram}
\end{figure}

\subsubsection{Optimization on Two Wave Beacons}
\label{sec:wpt_analysis}

Name the time length that all the required {\anbe}s harvest enough energy for a round of beacon operation as charging period (CP). 
For two wave beacons, 
the CP of Wave-I equals the charging time of the border {\anbe}, and the CP of Wave-II equals the charging time of the furthest {\anbe}. 
To understand how two wave beacons improve the energy utilization efficiency, we analyze the tradeoff between CP and PDA in this section. 
We first analyze two wave beacons in one dimension, and then explain the steps to extend the analysis to two dimensions.

\paragraph{Nodes Deployment}

We deploy \charbe, \mono, and {\anbe}s in a line. 
The deployment is shown in Fig.~\ref{fig:analysis_setup_twowave}. 
The {\charbe} is located at the beginning of the line. 
The {\mono} is randomly and uniformly distributed, such that \mono's distance $d$ from \charbe~is $d\sim\mathcal{U}(0,R)$, and $R$ is the maximum separation between \charbe~and \mono. 
The {\anbe} which is the closest to the {\charbe} is $d_v/2$ from the {\charbe}. 
Each \anbe~at location $i\in\{1,2,\ldots,r\}$ is positioned with a regular interval $d_v$, such that $d_vr=R$. 
Then the probability that the {\mono} is inside a cell is $\Omega = \frac{d_v}{R} = \frac{1}{r}$. 
The position of \anbe~$i$ is $d_i$. 
Assume each \anbe~$i$ has its beacon communication range, i.e. from $S^{(i)}_{\min}$, located at the middle distance between \anbe~ $i-1$ and \anbe~$i$, to $S^{(i)}_{\max}$ located at the middle distance between \anbe~$i$ and \anbe~$i+1$ (for $i=1$: position of \charbe). Finally, all \anbe s are divided into two groups, denoted as (i) first wave beacon group---with \anbe~$\{1,2,\ldots,w\}$, and (ii) second wave beacon group---with \anbe~$\{w,w+1,\ldots,r\}$. 
\anbe~$w$ is a border {\anbe}, which belongs to both groups. 

\begin{figure}
\centering
\subfigure[4 {\anbe}s.]{\includegraphics[height=0.40\columnwidth]{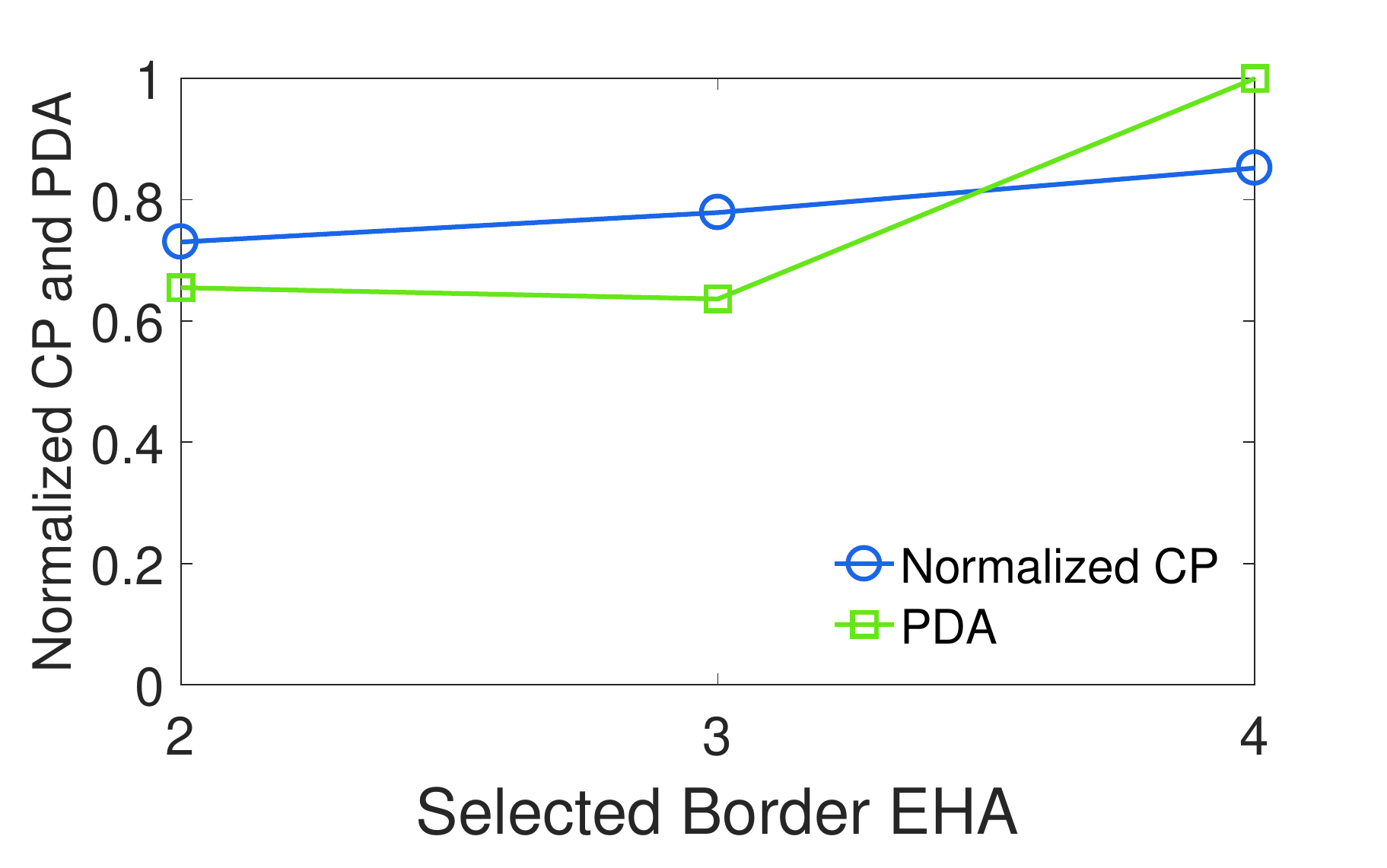}\label{sim:two_wave_line_4}}
\subfigure[10 {\anbe}s.]{\includegraphics[height=0.40\columnwidth]{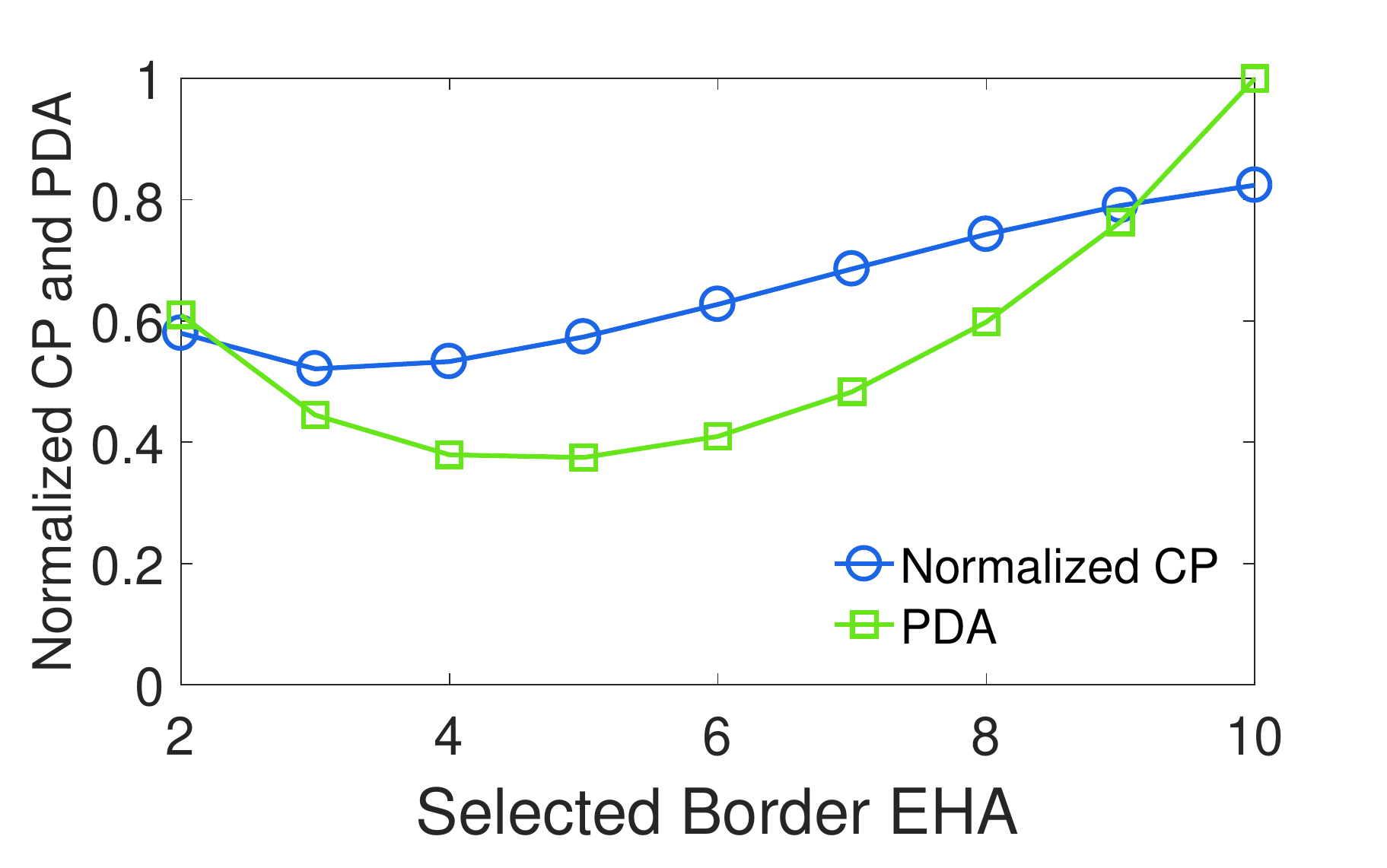}\label{sim:two_wave_line_10}}
\caption{Analysis of two wave beacons in one dimension. Each value of CP is normalized to the maximum CP in 4 {\anbe}s and 10 {\anbe}s respectively. 
The coordinate with label ``Selected Border {\anbe}'' represents the ID of border {\anbe}. 
The selected parameters: number of \anbe s (a) $r=4$, (b) $r=10$; \anbe~separation $d_v=1$; \anbe~transmission range $\forall i, S_{\min}^{(i)}=(i-1)d_v, S_{\max}^{(i)}=id_v$; log-normal shadowing parameter $\sigma=7$ (indoor environment); pathloss exponent $n=3$; reference distance $d_0=1$.}
\label{fig:simulation_two_wave_localization_line}
\end{figure}

\paragraph{Radio Propagation}

Assume the radio propagation follows a pathloss and log-normal shadowing model in which pathloss exponent is denoted as $n$, and log-normal random value $X\sim\mathcal{N}(0,\sigma)$. Denote $\varphi_y \triangleq P_t - L({d_0}) - 10n\log (y/{d_0})$ (dB) as the mean received signal strength at distance $y$, where $P_{t}$ is the \anbe~transmitted power, $L(d_0)$ is the (known in advance) reference pathloss at a reference distance $d_0\ll y$. The communication signal strength from \anbe~$i$ to \mono~is denoted as $P_{i} \sim N({\varphi _{|d-d_i|}},{\sigma ^2})$. Then $\Pr[{P_i} \geqslant {P_j}] \sim \mathcal{N}(\varphi_{|d-d_i|} - \varphi_{|d-d_j|},2{\sigma ^2})$, where $i>j$. Assume that $E_r$ is the energy required to perform one round of beacon by \anbe. Then $C(i)$ is the charging time that the {\anbe} $i$ needs to wait until it has enough energy for responding the next round of beacon request. Then $C_i=E_r/P_i$. 

\paragraph{Beacon Analysis}

Based on the above model, we are now ready to introduce the following proposition in one dimension of {\anbe}s. 
For two dimensional area, the analysis is the same except two points. Firstly, replace the spatial distribution $\Omega$ of {\mono} in the two dimensional area. Secondly, change the \emph{border {\anbe}} from one node $w$ to the set of border {\anbe}s. 

\begin{proposition}
The approach of two wave beacons sacrifices proximity detection accuracy of {\mono} for decreasing charging period of {\anbe}s. 
\end{proposition}

\begin{proof} 
Denote $K_i^{(j)}$ as the event that \mono~is in cell $i$ and the proximity detection results is in cell $j$. Suppose we only consider the capture effect in collision based beacon, with all $r$ {\anbe}s involved, the correct probability to decode beacon is calculated as
\begin{align}
K_i^{(j=i)} & = \max P_{i,\forall i \in \{1,r\}}\nonumber\\
 & = \Pr [{P_i} \geqslant \max({P_j}|j \in \{1,r\},j \ne i)].
\label{eq:K_i}
\end{align}

To simplify the calculation of (\ref{eq:K_i}) we use $K_i^{(j=i)}=\prod\nolimits_{j = 1, j \ne i}^r {\Pr[{P_i} \geqslant {P_j}]}$. 
From (\ref{eq:K_i}), the expected probability that \mono~is inside the correct cell $i$ is $H_i^{(i)} = \int_{S_{\min}^{(i)}}^{S_{\max}^{(i)}} K_i^{(j=i)} \mathrm{d}d$.

Now, if the first wave beacon group of the {\anbe}s is used, following the same analysis process, the event that \mono~is in cell $i$ while it is detected in cell $w$ is $K_i^{(j=w)} = \prod\nolimits_{j = 1}^{w-1} {\Pr[{P_w} \geqslant {P_j}]}$. The expected probability that \mono~is inside cell $i$ and it is detected in cell $w$ is $B_i^{(w)} = \int_{S_{\min }^{(i)}}^{S_{\max }^{(i)}} {K_i^{(j=w)}}\mathrm{d}d$.

If all {\anbe}s are woken up at the same time to beacon, then the expectation value of PDA and CP are
\begin{equation}
E_{a}^{a} = \Omega \sum\limits_{i = 1}^r H_i^{(i)}, 
\end{equation}
\begin{equation}
E_{c}^{a} = C(r). 
\end{equation}

If the two wave beacons is used, then the expectation value of PDA and CP are
\begin{equation}
E_{a}^{w} = \Omega \sum\limits_{i = 1}^{w-1} H_i^{(i)} + \Omega \sum\limits_{i = w}^r B_i^{(w)}H_i^{(i)},
\end{equation}
\begin{equation}
E_{c}^{w} = \Omega C(w)\sum\limits_{i = 1}^{r} (1 - B_i^{(w)}) + \Omega C(r)\sum\limits_{i = 1}^r B_i^{(w)}. 
\end{equation}

Due to the exponential decay of RF signal, we have $C(w) < C(r)$. 
Based on the above equations, it is easy to see that $E_{a}^{w} < E_{a}^{a}$ and $E_{c}^{w} < E_{c}^{a}$, which concludes the proof. 
\end{proof}

\begin{algorithm}[t]
\footnotesize
\caption{\small {\name} using \emph{Range Estimation}}
\scriptsize
\textbf{$\bullet$ {\charbe}}:
\begin{algorithmic} [1]
\label{alg2:1}
\Loop 
	\State Categorize {\anbe}s into two ranges $\Delta_{c}$ and $\Delta_{f}$. \Comment{See Sec.\ref{sec:distance_estimation}}
	\If{\texttt{wakeup-request} received}
		\State Send \texttt{wakeup-reply}
		\State Switch OFF and ON {\charbe}-Transmitter as passive wakeup signal.
		\State Evaluate $\xi$ in $\Delta_{c}$ or $\Delta_{f}$. 
		\If {\texttt{wakeup-request} has NO ``\textit{request-again}''}
			\State Broadcast \texttt{sleep} command to {\anbe}s in $\Delta$. 
		\Else
			\State Broadcast \texttt{sleep} command to {\anbe}s in $\overline \Delta$. 
		\EndIf
	\EndIf
\EndLoop 
\end{algorithmic}

\textbf{$\bullet$ {\anbe}}:
\begin{algorithmic} [1]
\label{alg2:3}
\Loop 
	\State Wakeup from sleep to monitor $D_{\text{out}}$ every $t_c$.
	\If{passive wakeup signal detected} \Comment{See Sec.\ref{sec:passive_wakeup}}
		\State Start receiving. 
		\If{\texttt{Sleep} command received}
			\State Goto Sleep.
		\EndIf
		\If{\texttt{beacon-request} received} \Comment{See Sec.\ref{sec:collision_beacon}}
			\State Send \texttt{beacon-reply}.
		\EndIf
	\EndIf
\EndLoop 
\end{algorithmic}
\textbf{$\bullet$ {\mono}}:
\begin{algorithmic} [1]
\label{alg2:2}
\Loop 
	\If{timer $\geq t_m$}
		\State Broadcast \texttt{wakeup-request}.
		\If{\texttt{wakeup-reply} received}
			\State Wait for $t_c$. \Comment{See Sec.\ref{sec:opt_passive_wakeup}}
			\State Broadcast \texttt{beacon-request} to {\anbe}s.
			\If{\texttt{beacon-reply} received from {\anbe}s}
				\State Decode information in \texttt{beacon-reply}.
				\State Finish this round of beacon.
			\Else
				\State Broadcast \texttt{wakeup-request} with ``\textit{request-again}''.
			\EndIf				
		\EndIf
	\EndIf
\EndLoop 
\end{algorithmic}
\label{alg:wiploc_localization_protocol}
\end{algorithm}

\begin{algorithm}[t]
\footnotesize
\caption{\small {\name} using \emph{Two Wave Beacons}}
\scriptsize
\textbf{$\bullet$ {\charbe}}:
\begin{algorithmic} [1]
\label{alg2:1_2}
\Loop 
	\State Categorize {\anbe}s into two waves Wave-I and Wave-II. \Comment{See Sec.\ref{sec:two_wave_wakeup}}
	\If{\texttt{wakeup-request} received}
		\State Send \texttt{wakeup-reply}
		\State Switch OFF and ON {\charbe}-Transmitter as passive wakeup signal.
		\If {\texttt{wakeup-request} has NO ``\textit{request-again}''}
			\State Broadcast \texttt{sleep} command to {\anbe}s in Wave-I. 
		\Else
			\State Broadcast \texttt{sleep} command to {\anbe}s in Wave-II. 
		\EndIf
	\EndIf
\EndLoop 
\end{algorithmic}
\textbf{$\bullet$ {\anbe}}:
\begin{algorithmic} [1]
\label{alg2:3_2}
\Loop 
	\State Wakeup from sleep to monitor $D_{\text{out}}$ every $t_c$.
	\If{passive wakeup signal detected} \Comment{See Sec.\ref{sec:passive_wakeup}}
		\State Start receiving. 
		\If{\texttt{Sleep} command received}
			\State Goto Sleep.
		\EndIf
		\If{\texttt{beacon-request} received} \Comment{See Sec.\ref{sec:collision_beacon}}
			\State Send \texttt{beacon-reply}.
		\EndIf
	\EndIf
\EndLoop 
\end{algorithmic}
\textbf{$\bullet$ {\mono}}:
\begin{algorithmic} [1]
\label{alg2:2_2}
\Loop 
	\If{timer $\geq t_m$}
		\State Broadcast \texttt{wakeup-request}.
		\If{\texttt{wakeup-reply} received}
			\State Wait for $t_c$. \Comment{See Sec.\ref{sec:opt_passive_wakeup}}
			\State Broadcast \texttt{beacon-request} to {\anbe}s.
			\If{\texttt{beacon-reply} received from {\anbe}s}
				\State Decode information in \texttt{beacon-reply}.
				\If {\texttt{beacon-reply} is from border {\anbe}}
					\If {This round of beacon is Wave-I} 
						\State Broadcast \texttt{beacon-request} with ``\textit{request-again}''.
					\Else
						\State Finish this round of beacon. 
					\EndIf
				\Else
					\State Finish this round of beacon. 
				\EndIf
			\Else
				\State Broadcast \texttt{wakeup-request} with ``\textit{request-again}''.
			\EndIf
		\EndIf
	\EndIf
\EndLoop 
\end{algorithmic}
\label{alg:wiploc_localization_protocol_2}
\end{algorithm}

We confirm this conclusion by the calculation of 4 {\anbe}s and 10 {\anbe}s respectively. The calculation results are shown in Fig.~\ref{fig:simulation_two_wave_localization_line}. 
If the border {\anbe} ID equals the total number of {\anbe}s (border {\anbe} = 4 in Fig.\ref{sim:two_wave_line_4}, border {\anbe} = 10 in Fig.\ref{sim:two_wave_line_10}), we do not use two wave beacons but wake up all {\anbe}s at once. 
Compared with the approach to wakeup all {\anbe}s, \emph{two wave beacons} approach loses some accuracy, however its average charging period decreases. 
We find that there is the minimum PDA, where border {\anbe} = 3 in Fig.\ref{sim:two_wave_line_4} and border {\anbe} = 5 in Fig.\ref{sim:two_wave_line_10}. 
We should avoid selecting these border {\anbe} while using the \emph{two wave beacons} approach. 
We further verify the above conclusion using the real hardware experiments in Sec.~\ref{sec:results}. 

\section{System Implementation}
\label{sec:system_implementation}

In this section, we explain the implementation of software and hardware based on the components of Sec.\ref{sec:key_tech}. 

\subsection{Software Implementation}
\label{sec:software_implementation}

The main beacon packets in {\name} are \texttt{beacon-request}, \texttt{beacon-reply}, \texttt{wakeup-request}, \texttt{wakeup-reply}, and \texttt{sleep}. All these packets have a fixed payload length of 30 bytes. 
The \texttt{beacon-request} and \texttt{wakeup-request} packets have the same format. They contain the ID of {\mono} in the first two bytes of the payload. 
The \texttt{beacon-reply} and \texttt{wakeup-reply} packets have the same format. They contains the {\anbe} ID encoded by the FEC which is then encoded by the orthogonal codes in the payload. 
For comparison, we present the program flow using range estimation (Sec.\ref{sec:distance_estimation}) and two wave beacons (Sec.\ref{sec:two_wave_wakeup}) in Algorithm~\ref{alg:wiploc_localization_protocol} and Algorithm~\ref{alg:wiploc_localization_protocol_2} respectively. 

\subsection{Hardware Implementation}
\label{sec:hardware_implementation}

Three types of hardware devices are implemented, including {\charbe}, {\anbe}, and {\mono}. 
The electrical connections of the three nodes are given in Fig.~\ref{fig:schematic_complete}.

\subsubsection{\textbf{\charbe}}
{\charbe} consists of two parts {\charbe}-Transmitter and {\charbe}-Controller. {\charbe}-Transmitter is Powercast TX91501 Powercaster transmitter~\cite{tx91501_p2110}, which has an \ac{EIRP} of 3\,W and operates at a center frequency of 915\,MHz. 
{\charbe}-Controller is the nRF51822 SoC Smart Beacon Kit~\cite{nrf51822_beacon_product_brief} with integrated PCB antenna, which is connected with {\charbe}-Transmitter. 
A transistor is added to the power line of the {\charbe}-Transmitter, so that the nRF51822 can switch the power of {\charbe}-Transmitter on or off. {\charbe}-Controller transmission power is 4\,dBm. The schematic of {\charbe} is shown in Fig.~\ref{fig:charbe_schematic}. 
The {\charbe} has a fixed power supply. Therefore {\charbe}-Controller listens all the time to the communication channel. 

When a high voltage signal is given to the transistor, {\charbe}-Transmitter is turned on and starts transmitting energy. After receiving a \texttt{wakeup-request} signal, 
{\charbe}-Controller controls {\charbe}-Transmitter to send a passive wakeup signal. 
The nRF51822 mote of {\charbe}-Controller pulls the transistor low to turn off the {\charbe}-Transmitter, after which it pulls the transistor high to turn on the {\charbe}-Transmitter again. 
By this way all the nearby {\anbe}s will detect the passive wakeup signal.

\subsubsection{\textbf{\anbe}}
The Powercast P2110B power harverster~\cite{tx91501_p2110} is used for harvesting power for the nRF51822 BLE mote. 
The harvester antenna is a vertical polarized patch antenna with a 6.1\,dBi gain. 
The version of development board for the nRF51822 is the Smart Beacon Kit. 
To keep the power consumption at a minimum, the transmission power of {\anbe} is set to -20\,dBm, and all the peripherals of the BLE mote are turned off except for one hardware timer which is set to generate an interrupt periodically. 
The CPU is most of the time in wait for interrupt state, which we denote as the \emph{sleep} state. 

The schematic of {\anbe} is shown in Fig.\ref{fig:anbe_schematic}. 
As {\anbe} is battery-less, the nRF51822 chip receives power from the $V_{\text{out}}$ pin of the harvester. The other connections are used for the passive wakeup and implemented as follows. The $D_{\text{set}}$ and $D_{\text{out}}$ pins are used for obtaining the received signal strength indication (RSSI). 
The $D_{\text{out}}$ and the $D_{\text{set}}$ are connected to P0.01 and P0.02 of the nRF51822, respectively. 
When $D_{\text{set}}$ is pulled high, $D_{\text{out}}$ represents the RSSI of the harvester. When the $D_{\text{out}}$ voltage is lower than the threshold $V_{t}$, the nRF51822 is woken up and starts listening to incoming radio packets until a valid packet is received. When a timeout is reached, the nRF51822 goes back to sleep.

\subsubsection{\textbf{\mono}}
Compared with the hardware and connections of {\anbe}, 
we only change the type of BLE board in {\mono} to increase the signal sensitivity. 
The BLE development board of {\mono} is nRF51822 nRFgo PCA10005~\cite{nrf51822_nRFgo}: an ARM Cortex M0 CPU with a helical monopole SMA-connected antenna. The transmission power is 4\,dBm. The schematic of \mono~is given in Fig.~\ref{fig:anbe_schematic}. 

\begin{figure}
	\centering
	\subfigure[{\charbe}.]{
		\includegraphics[width=0.6\columnwidth]{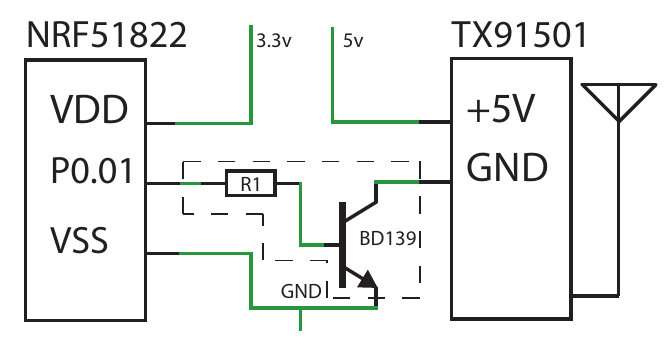}
		\label{fig:charbe_schematic}
	}	
	\subfigure[{\anbe} and {\mono}.]{
		\includegraphics[width=0.45\columnwidth]{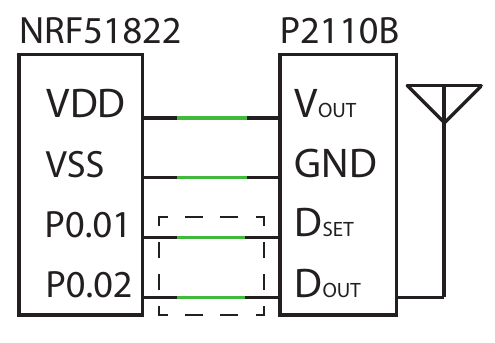}
		\label{fig:anbe_schematic}
	}
	\caption{Implementation of {\name} devices: (a) \charbe: Nordic Semiconductors nRF51822 SoC~\cite{nrf51822} connected to Powercast power transmitter~\cite{tx91501_p2110}. $R_1=330$\,$\Omega$; and (b) \anbe ~and \mono: nRF51822 connected to Powercast energy harvester P2110B~\cite{tx91501_p2110}. 
}
	\label{fig:schematic_complete}
\end{figure}

\section{Experimental Results}
\label{sec:results}

We evaluate the performance of {\name} from the following aspects. 
Firstly, we test whether the basic functions of {\name}, including BLE beacon and passive wakeup, can work properly using harvested RF energy (Sec.\ref{sec:benchmark_charging_time}). 
Secondly, we test the performance of \emph{collision based beacon} in an indoor area (Sec.\ref{sec:collision_beacon_room}). 
Thirdly, we test {\name} with increased deployment density of {\anbe}s. In this test, the components including \emph{collision based beacon}, \emph{passive wakeup}, and \emph{range estimation}. The operations of EHA and MONO completely rely on harvest RF energy (Sec.\ref{sec:test_high_density}). 
Finally, we test {\name} using \emph{two wave beacons} (Sec.\ref{sec:test_two_wave_beacons}), compared with the third experiment using \emph{range estimation}. 

\subsection{Simple Beacon with Passive Wakeup}
\label{sec:benchmark_charging_time}

To evaluate the performance of {\name}, it is necessary to know whether the implemented {\name} system is able to harvest enough energy for beacon communication from the RF energy transmitter. 
To achieve this aim, we measure the time length that {\anbe} needs to harvest RF energy for operating a beacon communication. 
In this experiment, 
we only test the basic beacon operation and passive wakeup with RF energy harvesting. 
There is not collision based beacon as in Sec.\ref{sec:collision_beacon}, optimized $t_c$ for minimizing average power consumption as in Sec.\ref{sec:opt_passive_wakeup}, range estimation as in Sec.\ref{sec:distance_estimation}, or two wave beacons as in Sec.\ref{sec:two_wave_wakeup}. 

\subsubsection{Experiment Setup}

\begin{figure}
\centering
\includegraphics[width=0.85\columnwidth]{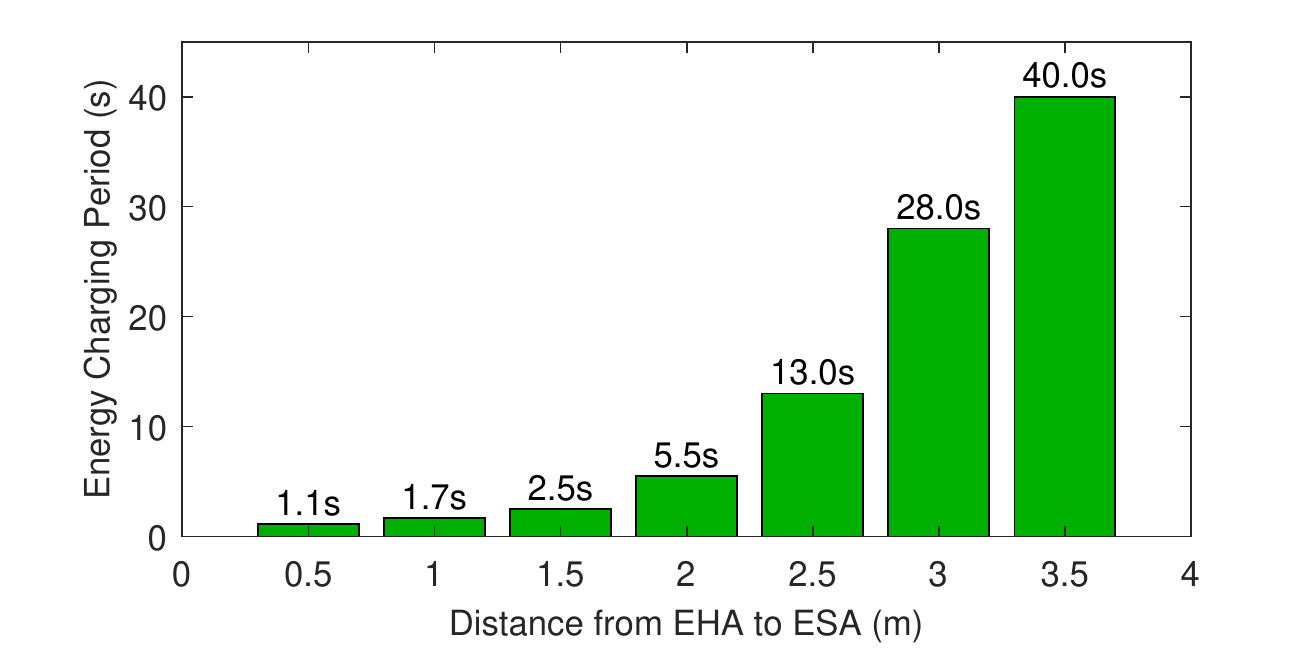}\label{exp:two_wave_ideal}
\caption{Relation between the distance from \anbe~to \charbe~and the best-effort energy charging period of {\anbe}. }
\label{fig:simple_beacon}
\end{figure}

One {\anbe} is used to send beacon messages. 
The test distance between {\charbe} and {\anbe} is from 0.5\,m to 3.5\,m with interval of 0.5\,m. 
A modified {\mono} is used in this test. 
The modified {\mono} connects to a personal computer, i.e. the harvester is removed while all the power is supported by the USB port of the PC. This modification is needed to keep the {\mono} in full-time listening mode. 
The {\anbe}, {\charbe}, and {\mono} are deployed in the communication range of each other. 

At each testing position, the {\mono} performs 100 \texttt{wakeup-request} with a fixed request period. This particular period is chosen starting at 500\,ms and increased with a step of 100\,ms. 
The {\anbe} is activated by passive wakeup signal of {\charbe} from sleeping mode. 
The {\anbe} sends a \texttt{beacon-reply} immediately after receiving the \texttt{beacon-request}. 
The \texttt{beacon-reply} packets from the {\anbe} are captured by the modified {\mono}. 
We measure the PRR. If the PRR is more than 95\%, we assume that the current energy charging period is long enough for the {\anbe} to finish a round of beacon operation at this particular distance. 

\subsubsection{Experiment Results}

The result of this experiment is presented in Fig.~\ref{fig:simple_beacon}. As expected the energy charging period increases as the distance increases. 
{\name} allows perpetual batteryless beacon using RF based energy harvesting, reaching beacon period of 40 s at a
3.5m distance between {\charbe} and {\anbe}. 
Further we infer that if one needs multiple {\anbe}s to work with the same period, all {\anbe}s need to wait as long as the charging period for the \anbe~that has the longest distance from {\charbe}. Therefore, the beacon response rate is low bounded by the furthest-located {\anbe}.

\subsection{Collision based Beacon}
\label{sec:collision_beacon_room}

In this experiment, we aim to evaluate the collision based beacon of Sec.\ref{sec:collision_beacon} in rooms. 
Therefore, we simplify workflow of {\name} to focus on collision beacons as follows. 
The {\charbe}-Transmitter is used for transmitting RF energy for the {\mono}. 
The passive wakeup in {\charbe}-Controller is not used. 
We remove the RF harvester of {\anbe} and connect to a wired power supply, so that the {\anbe} listens the communication channel in full time. If a \texttt{beacon-request} packet is received, the {\anbe} broadcasts \texttt{beacon-reply}. 
The {\mono} harvest RF energy as the power supply for beacon communication. 
It periodically wakes up from sleep to broadcasts \texttt{beacon-request}, and then its radio immediately switches to receiving mode. 
If a \texttt{beacon-reply} packet is received, it decode the message. 
After that, the {\mono} goes back to sleep mode. 

\subsubsection{Experiment Setup}

\begin{figure}  
\centering
\includegraphics[height=0.4\columnwidth]{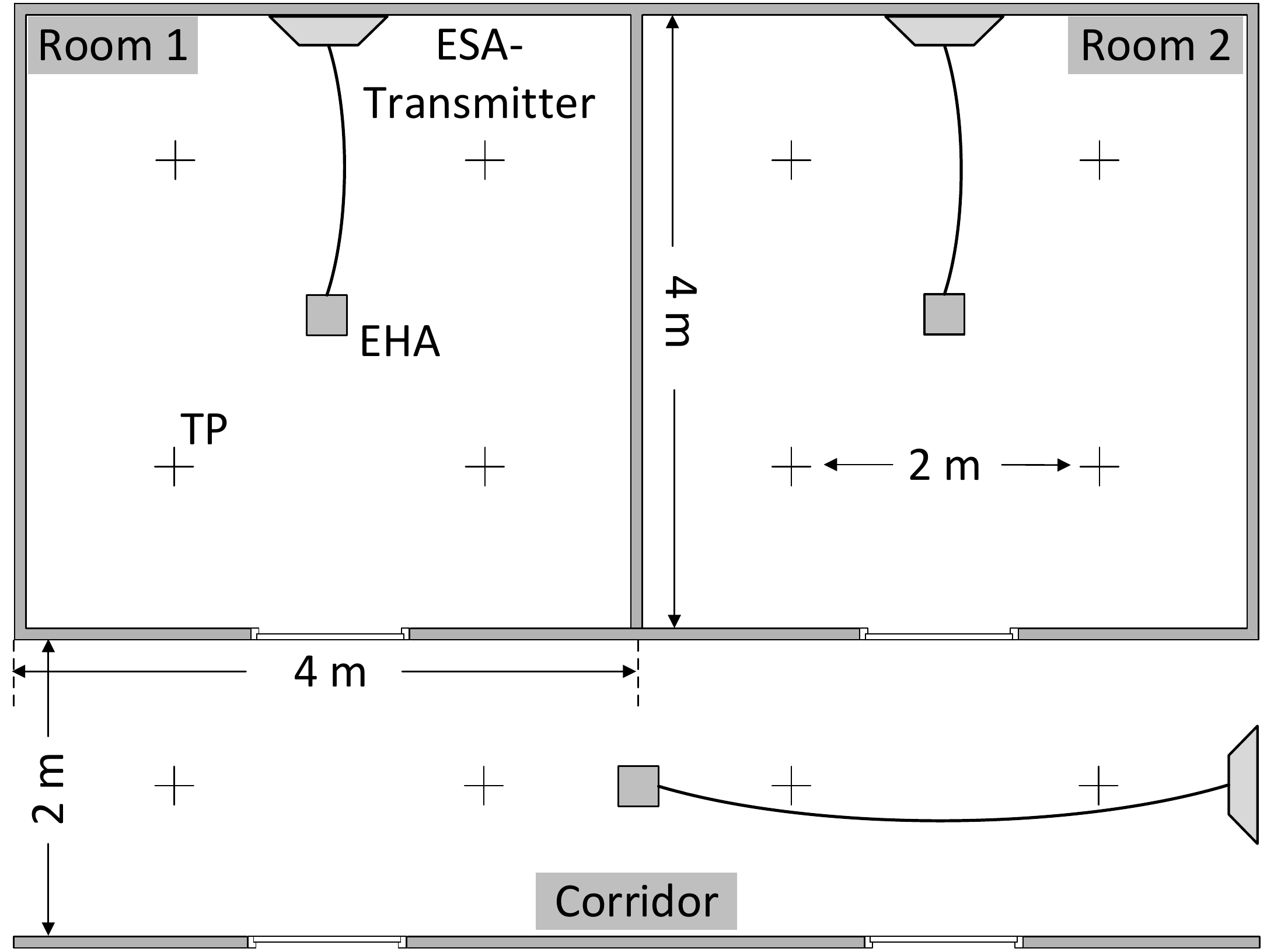}
\caption{Top view of \name~deployment in the office environment. {\anbe} has wired power supply connected to {\charbe}-Transmitter. The places marked as "$+$" are the testing positions (TP) of {\mono}.}
\label{fig:experiment_collision_beacon}
\end{figure}

\begin{table}[t]
	\centering
	\caption{Experiment result of collision based beacons}
	\label{tab:result1}
	\begin{tabular}{ c | c | c | c }
		\hline
		\textbf{Position of {\anbe}s} & \textbf{PRR (\%)} & \textbf{PDA (\%)} & \textbf{LE (m)}\\
		\hline
		\hline
		Room 1 and 2 & 95.5 & 99.3 & 0.03\\ \hline
		Room 1, 2 and Corridor & 84.7 & 89.6 & 0.25 \\ \hline
	\end{tabular}
\end{table}

The deployment of {\charbe}, {\anbe} and {\mono} is shown in Fig.\ref{fig:experiment_collision_beacon}. 
The room and corridor are divided into four cells of 2\,m$\times$2\,m respectively. 
The center of each cell is the test position of the {\mono}. 
Three {\anbe}s are deployed with wired power supply from {\charbe}-Transmitter in the center of room 1, room 2 and corridor respectively.  
The {\mono} is placed at every testing position. 
Every device is placed 1.0\,m above the floor and they are all in line-of-sight from each other. 

The {\mono} needs to consume extra energy for aggregating the testing results. 
As the {\mono} is powered wirelessly, consuming power for result data aggregation will negatively affect the performance of {\name}. 
We overcome this problem by using a BLE USB dongle as a sniffer~\cite{nrf51822}. 
This sniffer is deployed near the testing area, and it monitors all the packets sent by the {\mono}. 
In our experiment, the result of every beacon round is attached in the \texttt{beacon-request} packet of the next beacon round. 
The sniffer receives the data and saves it for further data processing.

\subsubsection{Experiment Results}

In the first round of experiment, we only deploy two {\anbe}s in the room 1 and 2. 
In the second round of experiment, we deploy all the three {\anbe}s in the rooms and corridor. 
At every test position, 50 beacon rounds were performed. 
For each test position, PRR, PDA and LE are computed and averaged.
The results are shown in Table~\ref{tab:result1}. 
We observe that the PRR and PDA decrease as the number of deployed {\anbe}s increases. 
This is mostly due the fact of multiple packets collision. 
In this experiment, we demonstrate two performance of {\name}. Firstly, we achieve batteryless {\mono} based on harvested RF energy. Secondly, the collision based beacon can achieve expected PRR and PDA, if we limit the number of {\anbe}s involved. 

\subsection{{\name} using Range Estimation in Cell Level {\anbe}s}
\label{sec:test_high_density}

\begin{figure}[t]
\centering
\includegraphics[height=0.75\columnwidth]{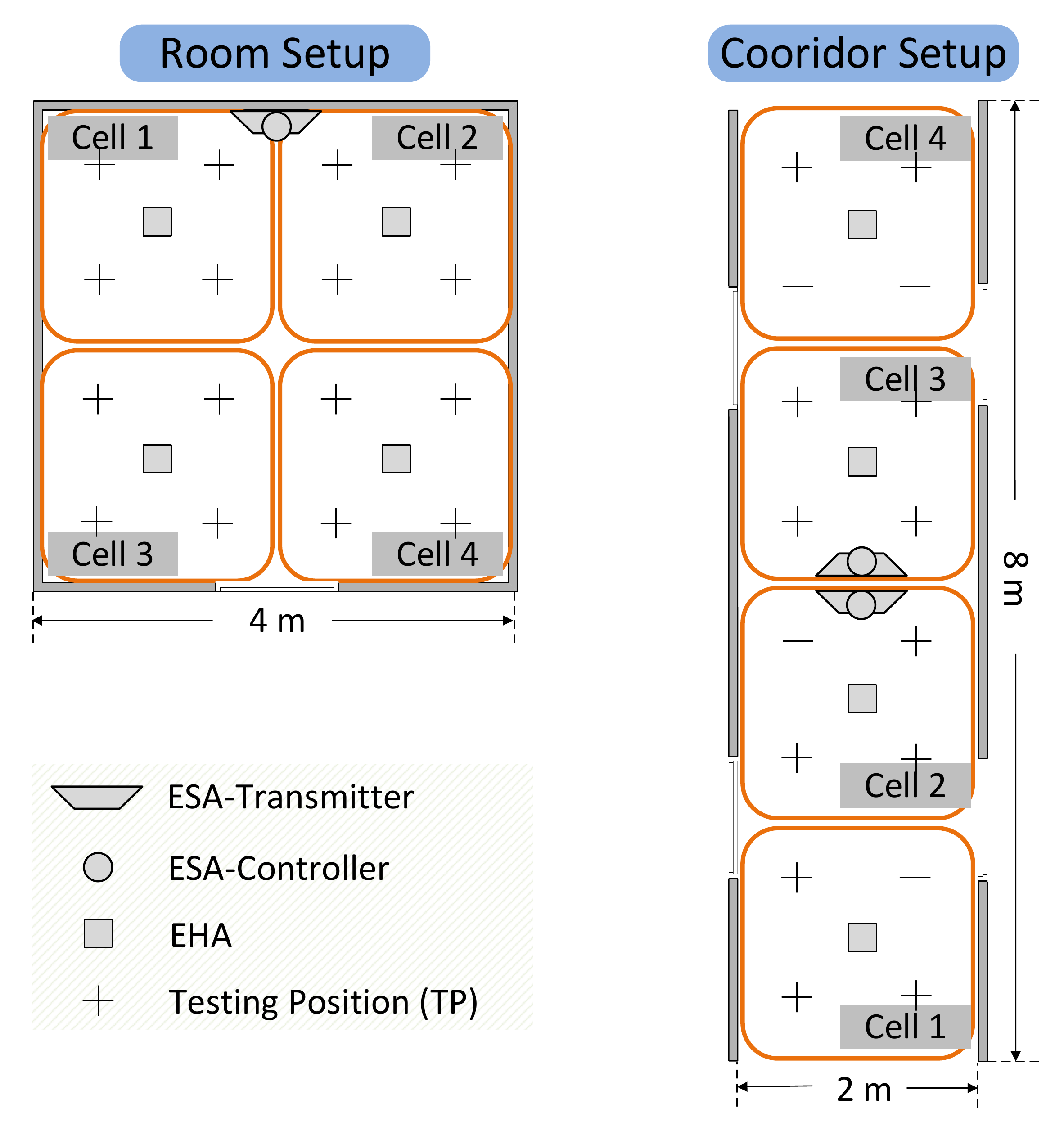}\label{fig:experiment_setup}
\caption{Experiment setup of the office and the corridor for cell-level testing of {\name} using range estimation. }
\label{fig:experiment_cell_level}
\end{figure}

In this experiment, we evaluate the performance of {\name} in cell-level of a real office area. 
We select the range estimation component, instead of the two wave beacons. 
The {\name} system is implemented under Alg.\ref{alg:wiploc_localization_protocol}. 
The deployment of {\anbe}s, {\charbe}s and the testing positions of {\mono} are illustrated in Fig.~\ref{fig:experiment_cell_level}. 

\subsubsection{Experiment Setup}

The experiments are performed in a room and a corridor respectively. 
Four {\anbe}s are deployed in each scenario. 
As the length of the corridor is longer than the effective energy transmitting range from {\charbe} to {\anbe}, we deploy two {\charbe}s back-to-back in the middle of the corridor pointing to the begin and the end of the corridor, respectively. 
Each {\charbe} in the corridor is responsible for waking up and controlling two {\anbe}s. 
We use the pre-measured harvested power value $\rho=3.7$\,dBm at the geographic middle position of the deployment area of {\anbe}s as the threshold value to categorize the two ranges. 
The {\mono} sends 20 beacon requests at each testing position. 
The same as in Sec.\ref{sec:collision_beacon_room}, 
an BLE packet sniffer was used for result data collection. 
The decoding for beacon packet is done by the {\mono}, so only the decoding result is transmitted to the sniffer. 

\subsubsection{Experiment Results}

\begin{table}[t]
	\centering
	\caption{Average experiment results in the room and the corridor of {\name} using range estimation.}
	\label{tab:experiment_cell_level}
	\begin{tabular}{ c | c | c | c  p{\columnwidth}}
	    \hline
		\textbf{Location} & \textbf{PRR (\%)} & \textbf{PDA (\%)} & \textbf{LE (m)}\\
		\hline
		\hline
		Room & 97.5 & 59.9 & 0.85 \\ \hline
		Corridor & 100 & 82.2 & 0.41 \\ \hline
	\end{tabular}
\end{table}

\begin{figure}[t]
 \centering	
 \subfigure[t][Cell-level testing results in the room.]{
 	\includegraphics[width=0.6\columnwidth]{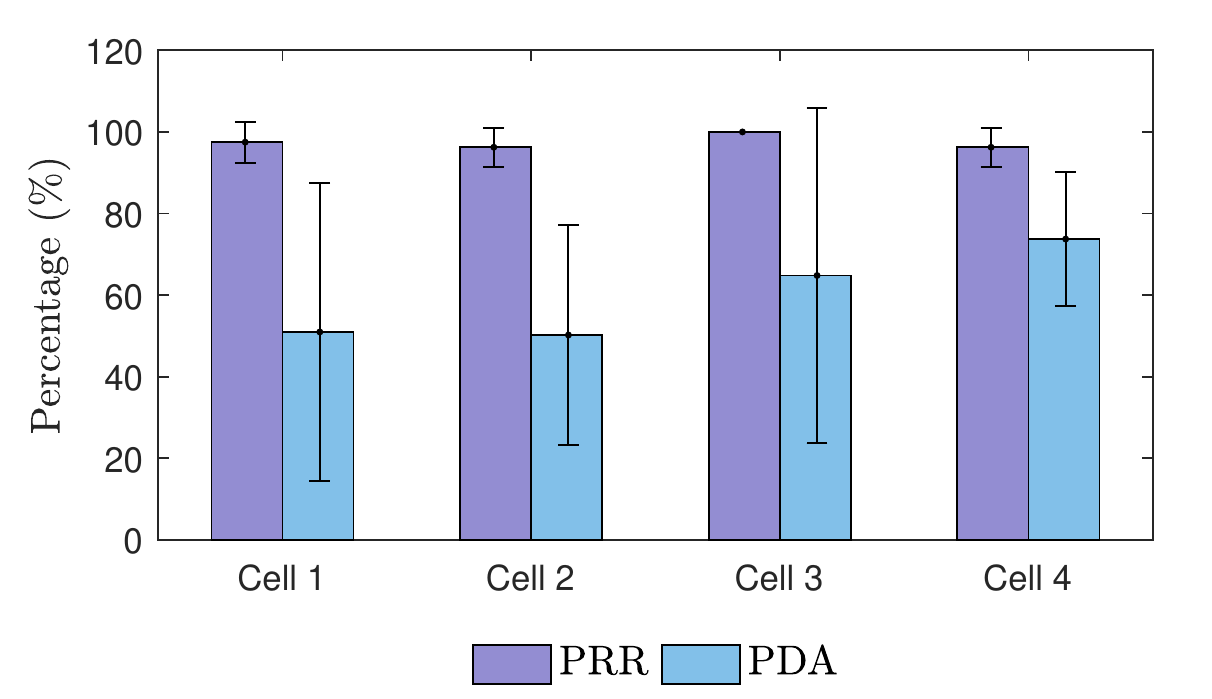}
 	\label{pic:cell_result_office}
 }
 \hspace{0.4cm}
 \subfigure[t][Cell-level testing results in the corridor.]{
 	\includegraphics[width=0.6\columnwidth]{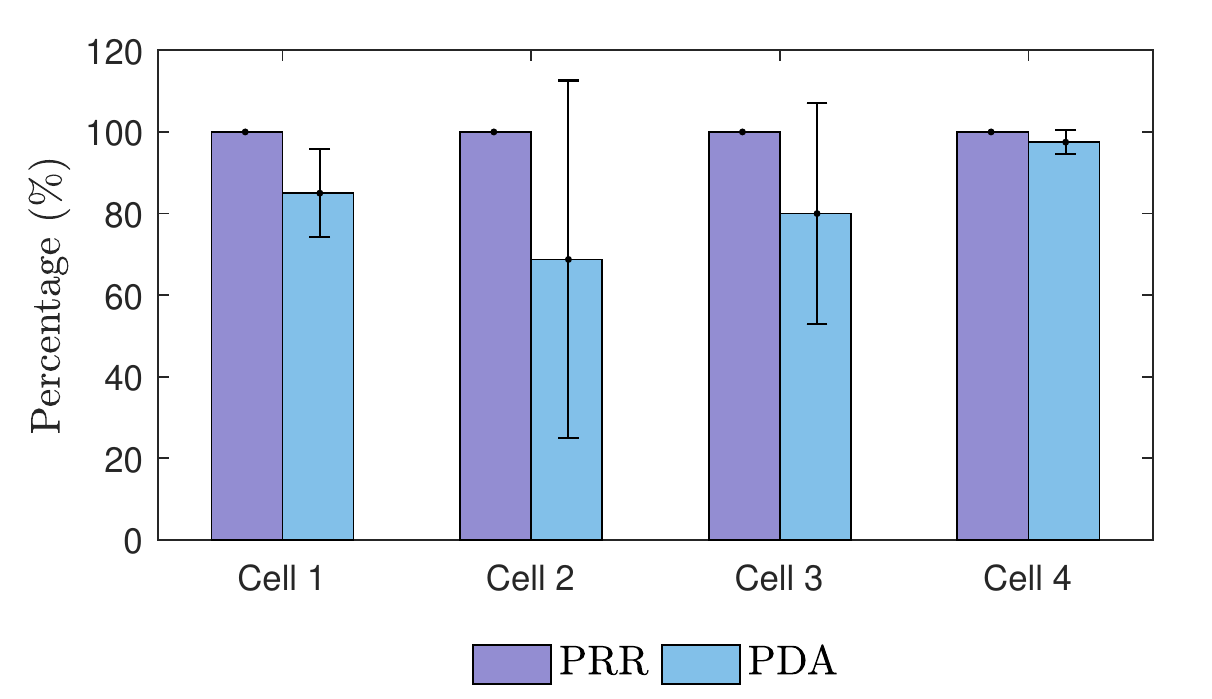}
  \label{pic:cell_result_corridor}
 }
 \subfigure[t][Localization error of cell-level testing results in the room.]{
 	\includegraphics[width=0.6\columnwidth]{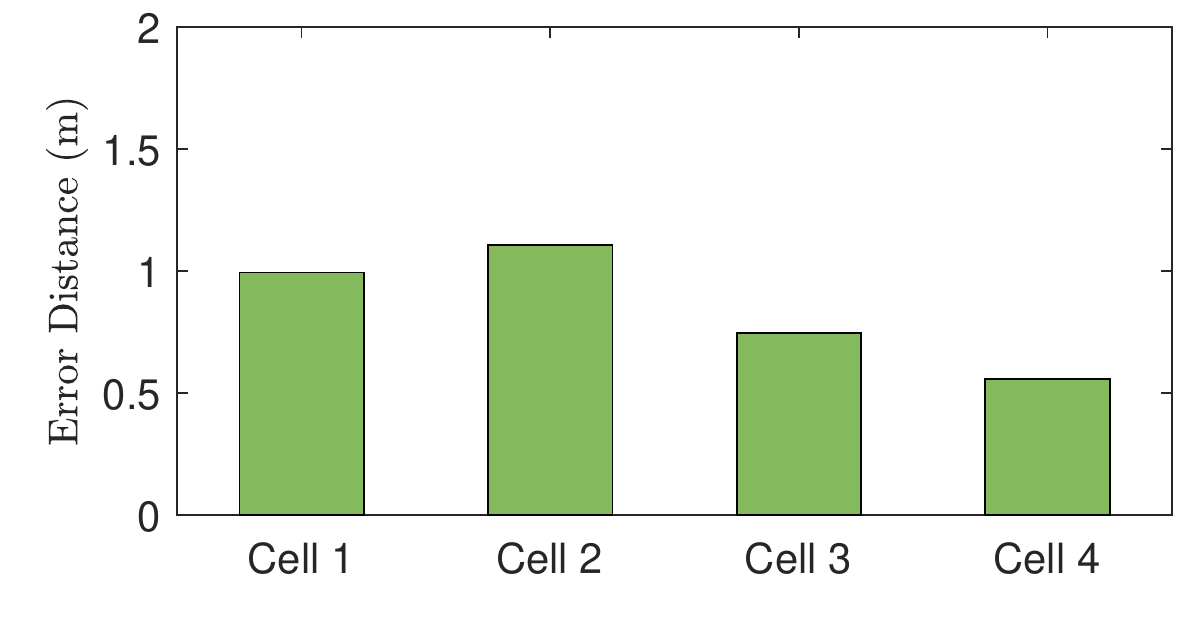}
 	\label{pic:cell_result_office_ed}
 }
 \hspace{0.4cm}
 \subfigure[t][Localization error of cell-level testing results in the corridor.]{
 	\includegraphics[width=0.6\columnwidth]{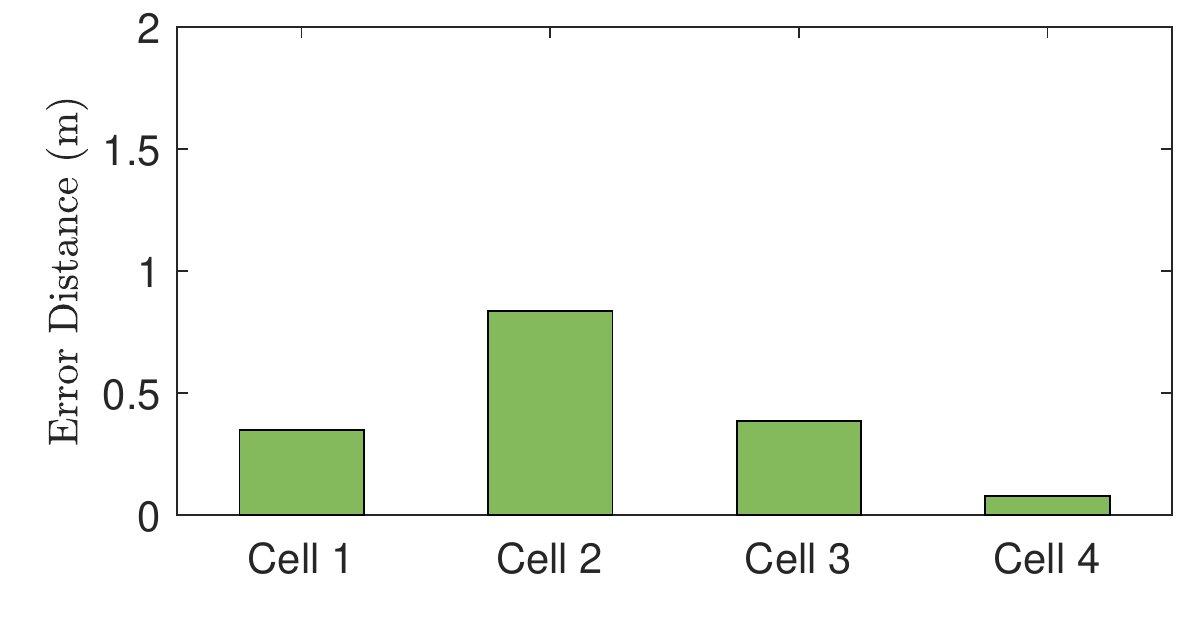}
  \label{pic:cell_result_corridor_ed}
 } 
 \caption{Experiment results of {\name} using range estimation in each cell of the room and the corridor.}
  \label{pic:cell_result_office_corridor}
\end{figure}

The average experiment results in the room and corridor are shown in Table~\ref{tab:experiment_cell_level}, and the results in each cell are shown in Fig.\ref{pic:cell_result_office_corridor}. 
Compared with the {\name} room-level experiment results as shown in Table~\ref{tab:result1}, the average PRR increases. 
This is because of the range estimation component, which selects only part of the {\anbe}s to send beacon signals in a beacon round. 
The cell-level accuracy of {\name} is lower than the room-level accuracy. 
We find that the cell-level accuracy at some cell positions is much lower than the others as shown in Fig.\ref{pic:cell_result_office_corridor}. 
This is due to two reasons: 

\begin{enumerate}[(i)]
\item The deployed cell area of each {\anbe} is only 4\,$\text{m}^\text{2}$ in this experiment, which is much smaller than the 16\,$\text{m}^\text{2}$ deployed room area in Sec.\ref{sec:collision_beacon_room}. 
The received beacon packets from {\anbe}s by the {\mono} only have small difference in RSS. This could cause error for decoding the packet from the closest {\anbe}. 
\item The radio pattern of {\charbe}-Transmitter has only 60$^\circ$ coverage in width and height, therefore some testing positions at the border of the room are not effectively covered. 
The {\mono} at these positions can not harvest enough energy and work properly. 
\end{enumerate}

The test results prove that {\name} is able to achieve cell-level beacon using harvested RF energy. 
The PRR is high enough for real applications. 
The PDA is low in some scenarios, which requires further improvement. 

\subsection{{\name} using Two Wave Beacons in Cell Level {\anbe}s}
\label{sec:test_two_wave_beacons}

In this experiment, we evaluate the performance of {\name} with two wave beacons component and validate that the two wave beacons mechanism improves beacon rate of {\anbe}s. 
We implement the {\name} system using Alg.\ref{alg:wiploc_localization_protocol_2}. 

\subsubsection{Experiment Setup}

\begin{figure}  
\centering
\subfigure[Top view of the deployment.]{\includegraphics[height=0.26\columnwidth]{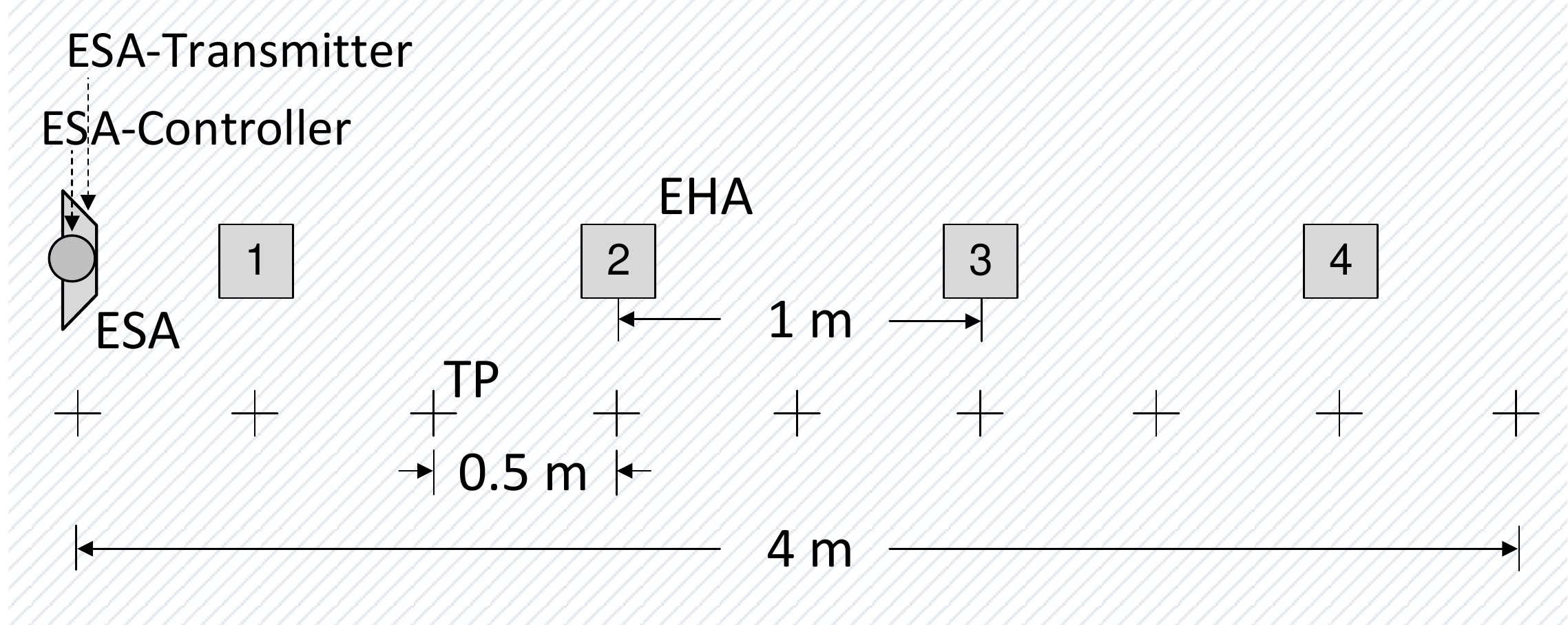}\label{fig:experiment_setup_cell}}
\hspace{0.1cm}
\subfigure[A picture of experiment setup in the corridor.]{\includegraphics[height=0.4\columnwidth]{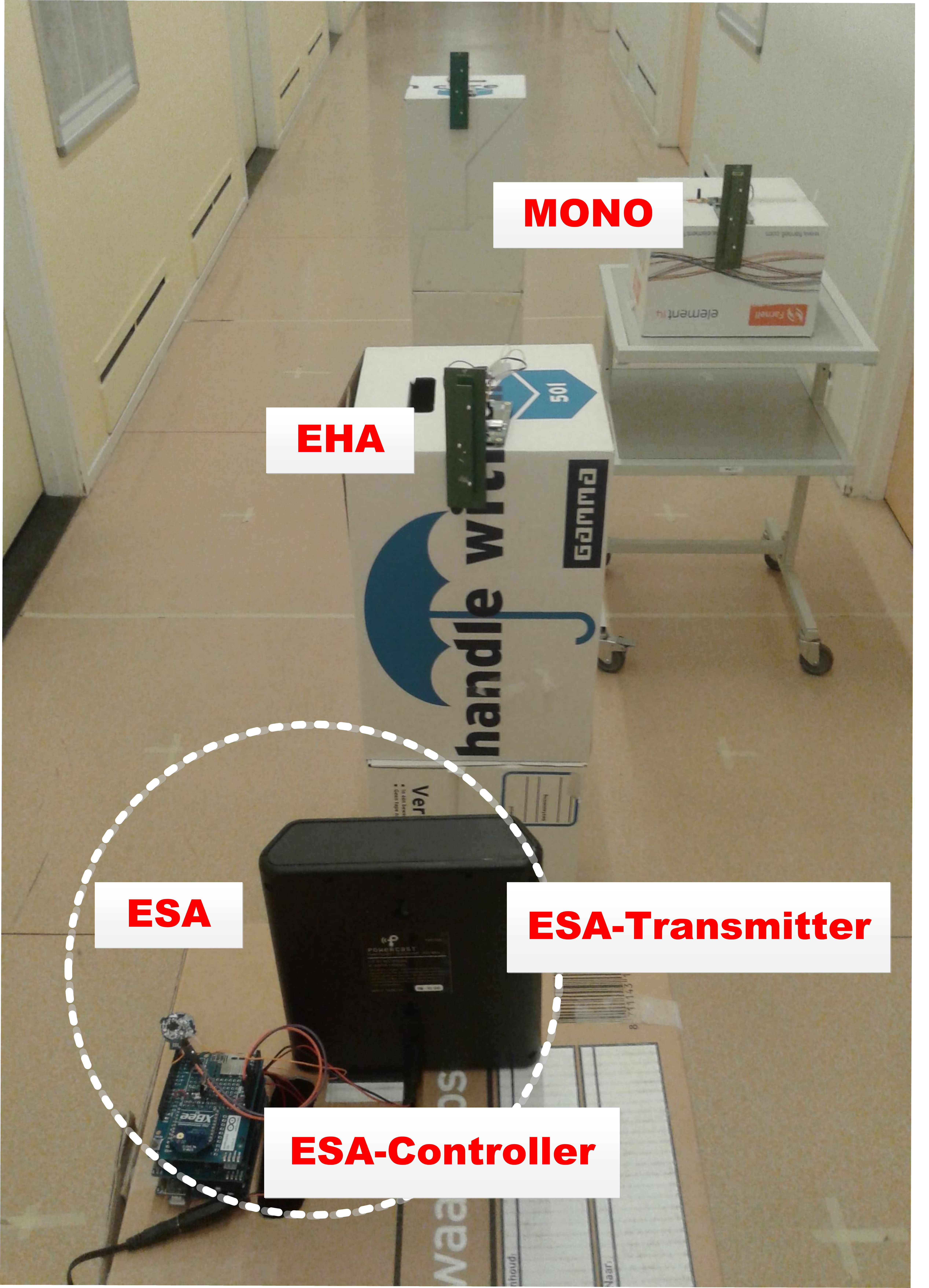}\label{fig:photo_setup}}
\caption{Experiment setup of {\name} using two wave beacons. The places marked as "$+$" are the testing positions (TP) of {\mono}.}
\label{fig:experiment_cell_level_twowave}
\end{figure}

We set up one {\charbe} and four {\anbe}s in a corridor, as depicted in Fig.~\ref{fig:experiment_cell_level_twowave}. 
Compared with the experiment in Sec.\ref{sec:test_high_density}, we deploy {\anbe}s closer to each other. 
The first {\anbe} is located at a 50\,cm distance from the \charbe~while the remaining three {\anbe}s are positioned at a 1\,m distance from the previous \anbe~on a straight line. The \mono~is located first at the same position as the \charbe~and moved away from it in steps of 50\,cm until a 4\,m distance is reached. At each {\mono}'s testing position 20 beacon events are performed. 
The procedure of packet decoding and results aggregation is the same as in Sec.\ref{sec:collision_beacon_room} and Sec.\ref{sec:test_high_density}. 
Three types of scenarios are tested for comparison: 

\begin{enumerate}
	\item \textbf{{\anbe} 2}: two wave beacons approach is enabled with border {\anbe} 2. 
	\item \textbf{{\anbe} 3}: two wave beacons approach is enabled with border {\anbe} 3. 
	\item \textbf{Wake all}: two wave beacons approach is disabled and all {\anbe}s broadcast \texttt{beacon-reply} after receiving \texttt{beacon-request}. 
\end{enumerate}

\subsubsection{Experiment Results}

The charging period of the {\anbe}s is defined and analyzed in Sec.\ref{sec:wpt_analysis}. 
The approach to measure the energy charging period of each {\anbe} is the same as in Sec.\ref{sec:benchmark_charging_time}. 
The experiment results on energy charging period of two wave beacons are shown in Fig.~\ref{fig:2wave_energy_harvesting_time}. 
Compared with the case when all {\anbe}s are woken up, the testing results with two wave beacons have smaller charging period. 
The average charging period with border {\anbe} 2 is around 47\% of ``wake all'' scenario. 
These results validate that the charging period of {\anbe}s can be reduced by two wave beacons approach. 
While decreasing the energy charging period, the experiment results in Fig.\ref{fig:2wave_prr_accuracy} illustrate the increasing of PRR and decreasing of PDA. 
The tradeoff between charging period and PDA aligns with the analysis results shown in Sec~\ref{sec:wpt_analysis}. 
Compared with the two wave beacons, the PRR in the scenario of ``wake all'' is lower. 
The main reason is that there are more collisions between the packets of {\anbe}s in the scenario of ``wake all''. 
Therefore, the probability that packets are corrupted increases. 
Meanwhile, compared with the scenario of ``wake all'', the accuracy of two wave beacons decreases. 
We analyze the testing results and find the following reason. 
In the first wave of beacons, while the {\mono} is outside the area of Wave-I, the {\mono} can not always receive and decode the ID of the border {\anbe}. 
This error causes that the second wave of beacons can not be activated. 

Overall we conclude that two wave beacons approach is able to decrease the average energy charging period of {\anbe}s and increase the PRR, while sacrificing some PDA. 

\begin{figure}[t]
\centering
\subfigure[Average energy charging period.]{\includegraphics[width=0.7\columnwidth]{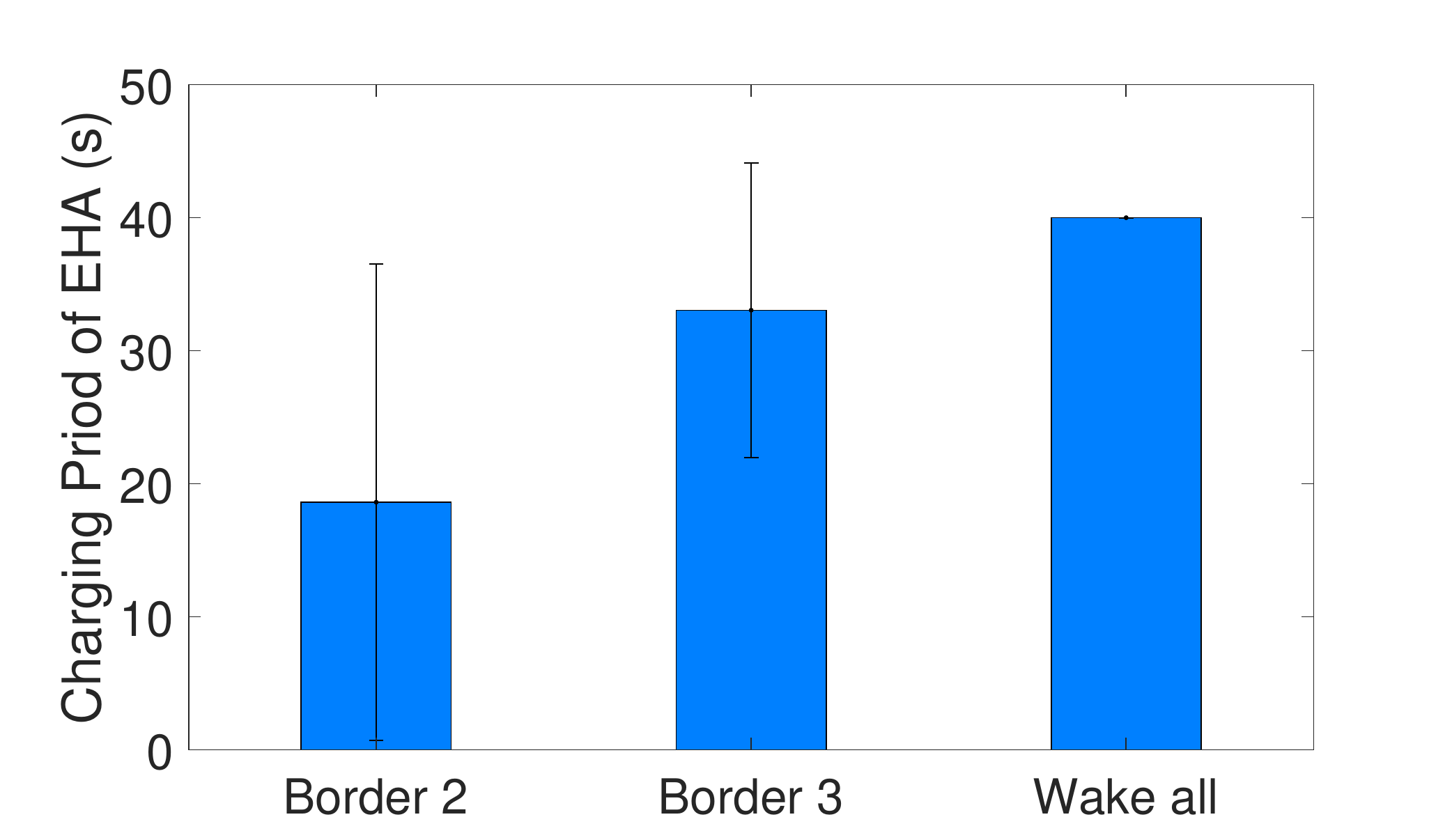}\label{fig:2wave_energy_harvesting_time}}

\subfigure[Average PRR and PDA.]{\includegraphics[width=0.7\columnwidth]{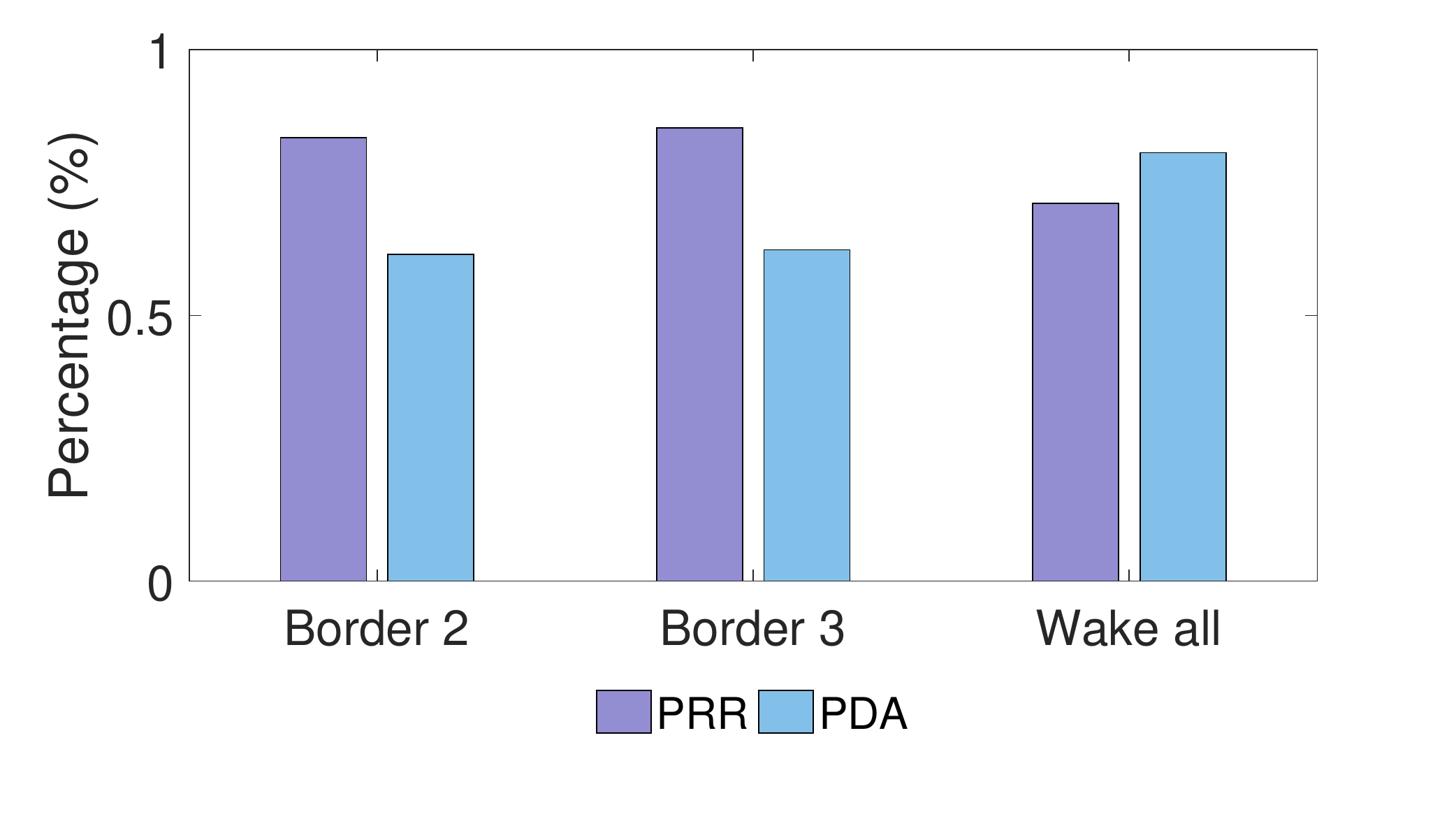}\label{fig:2wave_prr_accuracy}}

\subfigure[Localization error.]{\includegraphics[width=0.7\columnwidth]{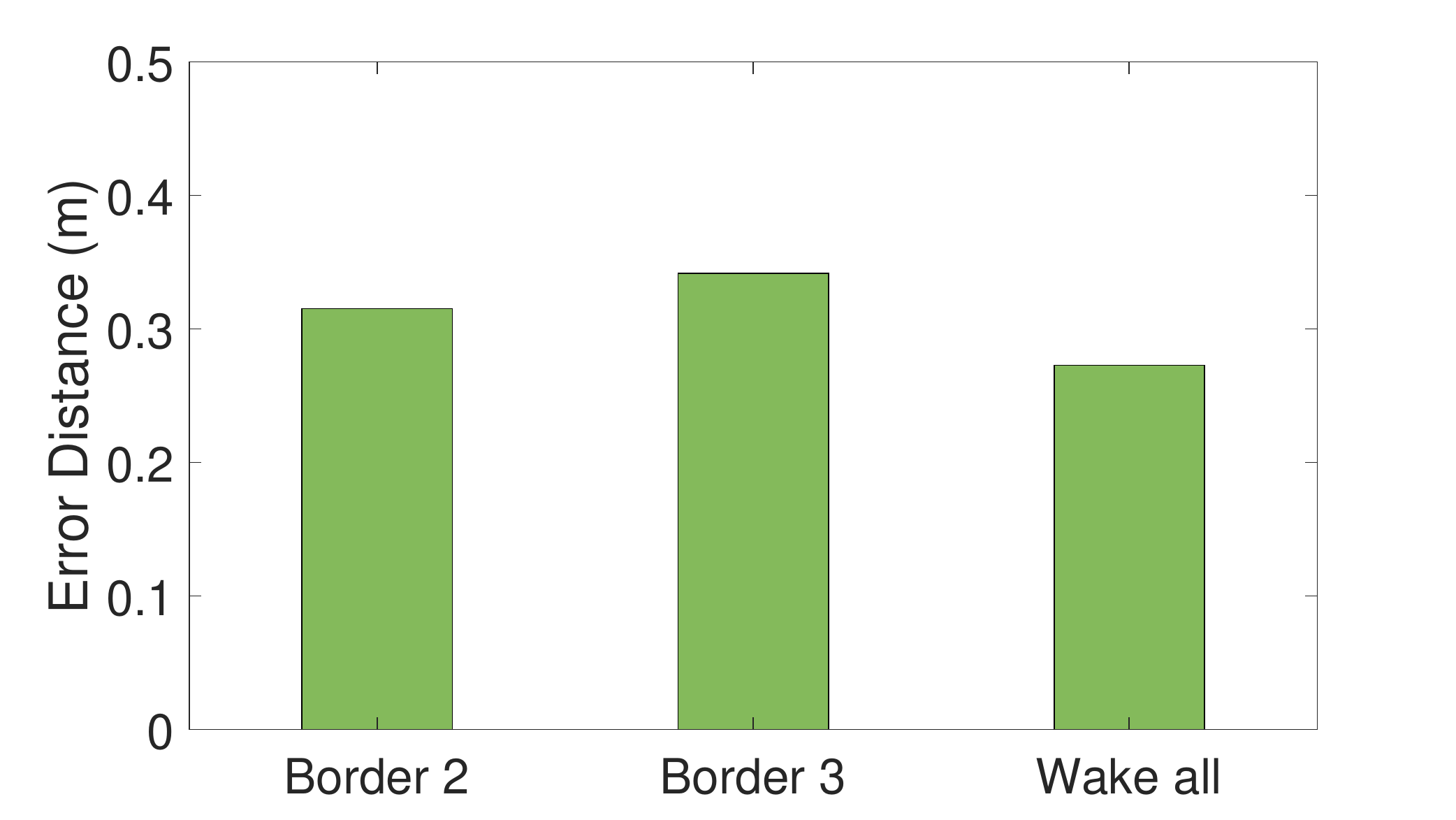}\label{fig:2wave_ed}}

\caption{Experiment results of {\name} using two wave beacons with different border {\anbe}s.}
\label{fig:two_wave_localization}
\end{figure}

\section{Future Work}
\label{sec:discussion}

Based on our experiment results, we conclude that the harvested RF energy is able to support beacon based applications in batteryless BLE devices. 
Meanwhile there exists some points worth of further improvement. 

\begin{enumerate}
\item {\name} does not consider the situations that affect the harvested RF power, such as obstacles between {\charbe}, {\anbe} and {\mono}. The system design assumes the devices are deployed in line-of-sight of each other. It is unknown how {\anbe}s should be categorized into two groups by range estimation approach or by two wave beacons approach, if the harvested energy of {\anbe}s are affected by nearby obstacles. We need an approach to cope with complicated energy radio environment. 

\item The performance of {\name} is unknown if multiple {\charbe}s are deployed. Firstly, the destructive radio from multiple RF energy transmitters will decrease the harvested energy of {\anbe}s and {\mono}s. Secondly, the passive wakeup signal from multiple {\charbe}s will trigger unnecessary wakeup of {\anbe}s.  To eliminate these effects, we need an approach to coordinate the operation of {\charbe}s. 

\item The scalability of collision based beacon in {\name} is unknown.  
Firstly, we only test the beacon communication with maximum 4 {\anbe}s. 
The testing results show that the PRR and PDA decrease as the number of {\anbe} increases. 
This means that collision based beacon approach should be further improved to cope with multiple beacon messages. 
A possible solution is using more efficient encoding and decoding process to reduce interference.  
Secondly, {\mono} needs to try all the orthogonal codes to decode the beacon message, which is energy-consuming. As the number of {\anbe} increases, the power consumption on decoding the orthogonal codes will become unbearable. 
We need a strategy to shrink the list of orthogonal codes. 
A possible solution is to divide the {\anbe} network into sub-areas. 
Instead of having an unique orthogonal code in the whole network, each {\anbe} has a pre-defined unique orthogonal code in each sub-area. 
Once a {\mono} is inside a sub-area, 
a {\charbe} sends the set of orthogonal codes of the sub-area to the {\mono}. 
\end{enumerate}

\section{Conclusion}
\label{sec:conclusion}

In this paper we presented a RF energy harvesting-enabled indoor BLE beacon system named as \textbf{\name}. 
The key innovations of {\name} are the four components to align the energy requirement of BLE communication to the limited budget of harvested energy, including: 
(i) leveraging collisions based beacon to extremely decrease energy used for receiving; 
(ii) using normal RF energy transmitting device to build a passive wakeup function, which is used for decreasing the power consumption on idle listening; 
(iii) waking up only the {\anbe}s near the {\mono} to send beacon messages by estimating the general range of the {\mono};
(iv) using two rounds (waves) of beacon to shift some workload from the {\anbe}s harvesting less energy to the ones harvesting more energy. 
We implement {\name} with these components using off-the-shelf BLE motes and RF harvesting devices. 
Based on the extensive indoor experiments, we show that {\name} is capable of providing continuous beacon services to mobile nodes based on limited harvested energy from dedicated RF transmitter. 
To the best of our knowledge, {\name} is the first indoor BLE beacon system powered by RF energy transmission. 

\bibliographystyle{IEEEtran}
\bibliography{IEEEabrv,wiploc_IEEE}

\end{document}